\documentstyle[12pt,fleqn]{article}

\font\twlgot =eufm10 scaled \magstep1 \font\egtgot =eufm8
\font\sevgot =eufm7

\font\twlmsb =msbm10 scaled \magstep1 \font\egtmsb =msbm8
\font\sevmsb =msbm7

\newfam\gotfam
\def\pgot{\fam\gotfam\twlgot}
\textfont\gotfam\twlgot \scriptfont\gotfam\egtgot
\scriptscriptfont\gotfam\sevgot
\def\got{\protect\pgot}

\newfam\msbfam
\textfont\msbfam\twlmsb \scriptfont\msbfam\egtmsb
\scriptscriptfont\msbfam\sevmsb
\def\Bbb{\protect\pBbb}
\def\pBbb{\relax\ifmmode\expandafter\Bb\else\typeout{You cann't use
Bbb in text mode}\fi}
\def\Bb #1{{\fam\msbfam\relax#1}}

\def\op#1{\mathop{{\it\fam0} #1}\limits}

\newcommand{\id}{{\rm Id\,}}

\newcommand{\pr}{{\rm pr}}

\newcommand{\di}{{\rm dim\,}}

\newcommand{\Id}{{\rm Id}}

\newcommand{\hm}{{\rm Hom\,}}
\newcommand{\dif}{{\rm Diff\,}}

\newcommand{\Is}{{\rm Aut\,}}
\newcommand{\nm}[1]{|{#1}|}

\newcommand{\bll}{\bullet}

\newcommand{\beq}{\begin{equation}}
\newcommand{\eeq}{\end{equation}}
\newcommand{\ben}{\begin{eqnarray}}
\newcommand{\een}{\end{eqnarray}}
\newcommand{\be}{\begin{eqnarray*}}
\newcommand{\ee}{\end{eqnarray*}}

\newcommand{\nw}[1]{[{#1}]}
\newcommand{\cB}{{\cal B}}

\newcommand{\gO}{{\got O}}
\newcommand{\gA}{{\got A}}

\newcommand{\cG}{{\got g}}
\newcommand{\gd}{{\got d}}

\newcommand{\gS}{{\got S}}

\newcommand{\gR}{{\got R}}

\newcommand{\cJ}{{\cal J}}
\newcommand{\cA}{{\cal A}}
\newcommand{\cO}{{\cal O}}

\newcommand{\cR}{{\cal R}}

\newcommand{\cV}{{\cal V}}

\newcommand{\cH}{{\cal H}}

\newcommand{\cC}{{\cal C}}

\newcommand{\cK}{{\cal K}}
\newcommand{\cM}{{\cal M}}

\newcommand{\ccG}{{\cal G}}

\newcommand{\cS}{{\cal S}}

\newcommand{\bL}{{\bf L}}

\newcommand{\bb}{{\bf 1}}

\newcommand{\al}{\alpha}

\newcommand{\dl}{\delta}
\newcommand{\la}{\lambda}
\newcommand{\La}{\Lambda}
\newcommand{\f}{\phi}
\newcommand{\vf}{\varphi}

\newcommand{\om}{\omega}

\newcommand{\g}{\gamma}
\newcommand{\G}{\Gamma}
\newcommand{\e}{\epsilon}
\newcommand{\ve}{\varepsilon}

\newcommand{\up}{\upsilon}
\newcommand{\vt}{\vartheta}

\newcommand{\si}{\sigma}
\newcommand{\Si}{\Sigma}

\newcommand{\w}{\wedge}

\newcommand{\wt}{\widetilde}
\newcommand{\wh}{\widehat}
\newcommand{\ol}{\overline}

\newcommand{\dr}{\partial}

\newcommand{\mar}[1]{}

\newcommand{\ar}{\op\longrightarrow}

\newcommand{\ot}{\otimes}

\let\ssection=\section
\renewcommand{\section}{\setcounter{equation}{0}\ssection}

\newcounter{eqalph}[section]
\newcounter{equationa}[section]
\newcounter{example}[section]
\newcounter{remark}[section]
\newcounter{theorem}[section]
\newcounter{proposition}[section]
\newcounter{lemma}[section]
\newcounter{corollary}[section]
\newcounter{definition}[section]

\def\theremark{\arabic{section}.\arabic{remark}}

\def\thedefinition{\arabic{section}.\arabic{definition}}

\newenvironment{proof}{\noindent {\bf Proof.}}{\hfill{\footnotesize\bf
QED} \bigskip }
\newenvironment{example}{\refstepcounter{remark} {\bf Example
\theremark.}}{{\Large $\bullet$}  }
\newenvironment{remark}{\refstepcounter{remark} {\bf Remark
\theremark.}}{{\Large $\bullet$}  }
\newenvironment{theorem}{\refstepcounter{definition} {\sc
Theorem \thedefinition}.}{$\Box$ }
\newenvironment{prop}{\refstepcounter{definition} {\sc
Proposition \thedefinition}.}{$\Box$  }

\newenvironment{definition}{\refstepcounter{definition} {\sc
Definition \thedefinition}.}{$\Box$ }

\begin{document}

\hbox{}

\begin{center}

{\large \bf Lectures on supergeometry}

\bigskip
\bigskip

{\sc G. Sardanashvily}
\bigskip

Department of Theoretical Physics, Moscow State University,
Moscow, Russia

\bigskip
\bigskip

{\bf Abstract}
\end{center}

\noindent Elements of supergeometry are an ingredient in many
contemporary classical and quantum field models involving odd
fields. For instance, this is the case of SUSY field theory, BRST
theory, supergravity. Addressing to theoreticians, these Lectures
aim to summarize the relevant material on supergeometry of modules
over graded commutative rings, graded manifolds and
supermanifolds.


\bigskip
\bigskip

\centerline{\bf Contents}
\bigskip

\noindent {\it 1. Graded tensor calculus} - {\bf 2}, {\it 2.
Graded differential calculus and connections} - {\bf 6}, {\it 3.
Geometry of graded manifolds} - {\bf 11}, {\it 4. Superfunctions}
- {\bf 18}, {\it 5. Supermanifolds} - {\bf 22}, {\it 6. DeWitt
supermanifolds} - {\bf 25}, {\it 7. Supervector bundles} - {\bf
26}, {\it 8. Superconnections} - {\bf 29}, {\it 9. Principal
superconnections} - {\bf 31}, {\it 10. Supermetric} - {\bf 36},
{\it 11. Graded principal bundles} - {\bf 40}.
\bigskip
\bigskip

Supergeometry is phrased in terms of $\Bbb Z_2$-graded modules and
sheaves over $\Bbb Z_2$-graded commutative algebras. Their
algebraic properties naturally generalize those of modules and
sheaves over commutative algebras, but supergeometry is not a
particular case of non-commutative geometry because of a different
definition of graded derivations. In these Lectures, we address
supergeometry of modules over graded commutative rings (Lecture
2), graded manifolds (Lectures 3 and 11) and supermanifolds.

It should be emphasized from the beginning that graded manifolds
are not supermanifolds, though every graded manifold determines a
DeWitt $H^\infty$-supermanifold, and {\it vice versa} (see Theorem
\ref{+44} below). Both graded manifolds and supermanifolds are
phrased in terms of sheaves of graded commutative algebras.
However, graded manifolds are characterized by sheaves on smooth
manifolds, while supermanifolds are constructed by gluing of
sheaves of supervector spaces. Note that there are different types
of supermanifolds; these are $H^\infty$-, $G^\infty$-,
$GH^\infty$-, $G$-, and DeWitt supermanifolds. For instance,
supervector bundles are defined in the category of
$G$-supermanifolds.

\section{Graded tensor calculus}

Unless otherwise stated, by a graded structure throughout the
Lectures is meant a $\Bbb Z_2$-graded structure, and the symbol
$[.]$ stands for the $\Bbb Z_2$-graded parity.

Let us recall some basic notions of the graded tensor calculus
\cite{bart,cia}.

An algebra $\cA$ is called graded if it is endowed with a grading
automorphism $\g$ such that $\g^2=\id$. A graded algebra seen as a
$\Bbb Z$-module falls into the direct sum $\cA=\cA_0\oplus \cA_1$
of two $\Bbb Z$-modules $\cA_0$ and $\cA_1$ of even and odd
elements such that
\be
\g(a)=(-1)^ia, \qquad a\in\cA_i, \qquad i=0,1.
\ee
One calls $\cA_0$ and $\cA_1$ the even and odd parts of $\cA$,
respectively. In particular, if $\g=\id$, then $\cA=\cA_0$. Since
\be
\g(aa')=\g(a)\g(a'),
\ee
we have
\be
[aa']=([a]+[a']){\rm mod}\,2
\ee
where $a\in \cA_{[a]}$, $a'\in \cA_{[a']}$. It follows that
$\cA_0$ is a subalgebra of $\cA$ and $\cA_1$ is an $\cA_0$-module.
If $\cA$ is a graded ring, then $[\bb]=0$.

A graded algebra $\cA$ is said to be graded commutative if
\be
aa'=(-1)^{[a][a']}a'a,
\ee
where $a$ and $a'$ are arbitrary homogeneous elements of $\cA$,
i.e., they are either even or odd.

Given a graded algebra $\cA$, a left graded $\cA$-module $Q$ is a
left $\cA$-module provided with the grading automorphism $\g$ such
that
\be
\g(aq)=\g(a)\g(q), \qquad a\in\cA,\qquad q\in Q,
\ee
i.e.,
\be
[aq]=([a]+[q]){\rm mod}\,2.
\ee
A graded module $Q$ is split into the direct sum $Q=Q_0\oplus Q_1$
of two $\cA_0$-modules $Q_0$ and $Q_1$ of even and odd elements.
Similarly, right graded modules are defined.

If $\cK$ is a graded commutative ring, a graded $\cK$-module can
be provided with a graded $\cK$-bimodule structure by letting
\be
qa= (-1)^{[a][q]}aq, \qquad a\in\cK, \qquad q\in Q.
\ee
A graded $\cK$-module is called free if it has a basis generated
by homogeneous elements. This basis is said to be of type $(n,m)$
if it contains $n$ even and $m$ odd elements.

In particular, by a (real) graded vector space $B=B_0\oplus B_1$
is meant a graded $\Bbb R$-module. A graded vector space is said
to be $(n,m)$-dimensional if $B_0=\Bbb R^n$ and $B_1=\Bbb R^m$.

The following are standard constructions of new graded modules
from old ones.

$\bullet$ The direct sum of graded modules over the same graded
commutative ring and a graded factor module are defined just as
those of modules over a commutative ring.

$\bullet$ The tensor product $P\ot Q$ of graded $\cK$-modules $P$
and $Q$ is an additive group generated by elements $p\ot q$, $p\in
P$, $q\in Q$, obeying the relations
\be
&& (p+p')\ot q =p\ot q + p'\ot q, \\
&& p\ot(q+q')=p\ot q+p\ot q', \\
&&  ap\ot q=(-1)^{[p][a]}pa\ot q= (-1)^{[p][a]}p\ot aq=\\
&& \qquad (-1)^{([p]+[q])[a]}p\ot qa, \qquad a\in\cK.
\ee
In particular, the tensor algebra
\be
\ot Q=\cK\oplus Q\oplus\cdots \oplus \op\ot^k Q\oplus \cdots
\ee
of a graded $\cK$-module $Q$ is defined as that of a module over a
commutative ring. Its quotient $\w Q$ with respect to the ideal
generated by elements
\be
q\ot q' + (-1)^{[q][q']}q'\ot q, \qquad q,q'\in Q,
\ee
is the bigraded exterior algebra of a graded module $Q$ with
respect to the graded exterior product
\be
q\w q' =- (-1)^{[q][q']}q'\w q.
\ee

$\bullet$ A morphism $\Phi:P\to Q$ of graded $\cK$-modules is said
to be even (resp. odd) if $\Phi$ preserves (resp. change) the
graded parity of all elements $P$. It obeys the relations
\be
\Phi(ap)=(-1)^{[\Phi][a]}\Phi(p), \qquad p\in P, \qquad a\in\cK.
\ee
The set $\hm_\cK(P,Q)$ of graded morphisms of a graded
$\cK$-module $P$ to a graded $\cK$-module $Q$ is naturally a
graded $\cK$-module. The graded $\cK$-module $P^*=\hm_\cK(P,\cK)$
is called the dual of a graded $\cK$-module $P$.

A graded commutative $\cK$-ring $A$ is a graded commutative ring
which is also a graded $\cK$-module. A graded commutative $\Bbb
R$-ring is said to be of rank $N$ if it is a free algebra
generated by the unit $\bb$ and $N$ odd elements. A graded
commutative Banach ring $A$ is a graded commutative $\Bbb R$-ring
which is a real Banach algebra whose norm obeys the additional
condition
\be
\|a_0 + a_1\|=\|a_0\| + \|a_1\|, \qquad a_0\in A_0, \quad a_1\in
A_1.
\ee

Let $V$ be a real vector space. Let $\La=\w V$ be its ($\Bbb
N$-graded) exterior algebra provided  with the $\Bbb Z_2$-graded
structure
\mar{+66}\beq
\La=\La_0\oplus \La_1, \qquad \La_0=\Bbb R\op\bigoplus_{k=1}
\op\w^{2k} V, \qquad \La_1=\op\bigoplus_{k=1} \op\w^{2k-1} V.
\label{+66}
\eeq
It is a graded commutative $\Bbb R$-ring, called the Grassmann
algebra. A Grassmann algebra, seen as an additive group, admits
the decomposition
\mar{+11}\beq
\La=\Bbb R\oplus R =\Bbb R\oplus R_0\oplus R_1=\Bbb R \oplus
(\La_1)^2 \oplus \La_1, \label{+11}
\eeq
where $R$ is the ideal of nilpotents of $\La$. The corresponding
projections $\si:\La\to\Bbb R$ and $s:\La\to R$ are called the
body and soul maps, respectively.

\begin{remark} \label{+52} \mar{+52}
Let us note that there is a different definition of a Grassmann
algebra \cite{jad} which is equivalent to the above one only in
the case of an infinite-dimensional vector space $V$ \cite{cia}.
Let us mention the Arens--Michael algebras of Grassmann origin
\cite{bruz99} which are most general graded commutative algebras,
suitable for superanalysis (see Remark \ref{+51} below).
\end{remark}

Hereafter, we restrict our consideration to Grassmann algebras of
finite rank. Given a basis $\{c^i\}$ for the vector space $V$, the
elements of the Grassmann algebra $\La$ (\ref{+66}) take the form
\mar{z784}\beq
a=\op\sum_{k=0} \op\sum_{(i_1\cdots i_k)}a_{i_1\cdots
i_k}c^{i_1}\cdots c^{i_k}, \label{z784}
\eeq
where the second sum runs through all the tuples $(i_1\cdots i_k)$
such that no two of them are permutations of each other. The
Grassmann algebra $\La$ becomes a graded commutative Banach ring
if its elements (\ref{z784}) are endowed with the norm
\be
\|a\|=\op\sum_{k=0} \op\sum_{(i_1\cdots i_k)}\nm{a_{i_1\cdots
i_k}}.
\ee

Let $B$ be a graded vector space. Given a Grassmann algebra $\La$
of rank $N$, it can be brought into a graded $\La$-module
\be
\La B=(\La B)_0\oplus (\La B)_1=(\La_0\ot B_0\oplus \La_1\ot
B_1)\oplus (\La_1\ot B_0\oplus \La_0\ot B_1),
\ee
called a superspace. The superspace
\mar{+70}\beq
B^{n\mid m}=[(\op\oplus^n\La_0) \oplus (\op\oplus^m\La_1)]\oplus
[(\op\oplus^n\La_1)\oplus (\op\oplus^m\La_0)] \label{+70}
\eeq
is said to be $(n,m)$-dimensional. The graded $\La_0$-module
\be
B^{n,m}= (\op\oplus^n\La_0) \oplus (\op\oplus^m\La_1)
\ee
is called an $(n,m)$-dimensional supervector space.

Whenever referring to a topology on a supervector space
$B^{n,m}$, we will mean the Euclidean topology on a
$2^{N-1}[n+m]$-dimensional real vector space.

Given a superspace $B^{n\mid m}$  over a Grassmann algebra $\La$,
a $\La$-module endomorphism of $B^{n\mid m}$ is represented by an
$(n+ m)\times (n+m)$ matrix
\mar{+200}\beq
L=\left(
\begin{array}{cc}
L_1 & L_2 \\
L_3 & L_4
\end{array}
\right) \label{+200}
\eeq
with entries in $\La$. It is called a supermatrix. One says that a
supermatrix $L$ is

$\bullet$ even if $L_1$ and $L_4$ have even entries, while $L_2$
and $L_3$ have the odd ones;

$\bullet$ odd if $L_1$ and $L_4$ have odd entries, while $L_2$ and
$L_3$ have the even ones.

Endowed with this gradation, the set of supermatrices (\ref{+200})
is a graded $\La$-ring. Unless otherwise stated, by supermatrices
are meant homogeneous ones.

The familiar notion of a trace is extended to supermatrices
(\ref{+200}) as the supertrace
\be
{\rm Str}\,L ={\rm Tr}\,L_1 -(-1)^{[L]}{\rm Tr}\,L_4.
\ee
For instance, Str$(\,\bb)=n-m$.

A supertransposition $L^{st}$ of a supermatrix $L$ is the
supermatrix
\be
L^{st}=\left(
\begin{array}{cc}
L_1^t & (-1)^{[L]}L_3^t \\
-(-1)^{[L]}L_2^t & L_4^t
\end{array}
\right),
\ee
where $L^t$ denotes the ordinary matrix transposition. There are
the relations
\mar{+262',2}\ben
&& {\rm Str}(L^{st})={\rm Str}\,L,\nonumber\\
&& (LL')^{st}=(-1)^{[L][L']}L'^{st}L^{st}, \label{+262'}\\
&&{\rm Str}(LL')=(-1)^{[L][L']}{\rm Str}(L'L) \quad {\rm or} \quad {\rm
Str}([L,L'])=0.\label{+262}
\een

In order to extend the notion of a determinant to supermatrices,
let us consider invertible supermatrices $L$ (\ref{+200}). They
are never odd. One can show that an even supermatrix $L$ is
invertible  if and only if either the matrices $L_1$ and $L_4$ are
invertible or the real matrix $\si(L)$ is invertible, where $\si$
is the body morphism.

Invertible supermatrices constitute a group $GL(n|m;\La)$, called
the general linear graded group. Then a superdeterminant of $L\in
GL(n|m;\La)$ is defined as
\be
{\rm Sdet}\,L={\rm det}(L_1 -L_2L_4^{-1}L_3)({\rm det}\,L^{-1}_4).
\ee
It satisfies the relations
\be
&& {\rm Sdet}(LL')=({\rm Sdet}\,L)({\rm Sdet}\,L'), \\
&& {\rm Sdet}(L^{st})= {\rm Sdet}\,L,\\
 && {\rm Sdet}(\exp(L))=\exp({\rm Sdet}\,(L)).
\ee

Let $\cK$ be a graded commutative ring. A graded commutative
(non-associative) $\cK$-algebra $\cG$ is called a Lie
$\cK$-superalgebra if its product, called the superbracket
 and denoted by $[.,.]$, obeys the relations
\be
&& [\ve,\ve']=-(-1)^{[\ve][\ve']}[\ve',\ve],\\
&& (-1)^{[\ve][\ve'']}[\ve,[\ve',\ve'']]
+(-1)^{[\ve'][\ve]}[\ve',[\ve'',\ve]] +
(-1)^{[\ve''][\ve']}[\ve'',[\ve,\ve']] =0.
\ee
Obviously, the even part $\cG_0$ of a Lie $\cK$-superalgebra $\cG$
is a Lie $\cK_0$-algebra. A graded $\cK$-module $P$ is called a
$\cG$-module if it is provided with a $\cK$-bilinear map
\be
&& \cG\times P\ni (\ve,p)\mapsto \ve p\in P, \\
&& [\ve
p]=([\ve]+[p]){\rm mod}\,2,\\
&& [\ve,\ve']p=(\ve\circ\ve'-(-1)^{[\ve][\ve']}\ve'\circ\ve)p.
\ee

\section{Graded differential calculus and connections}

Linear differential operators and connections on graded modules
over graded commutative rings are defined similarly to those in
commutative geometry \cite{book09,book00,sard09}.

Let $\cK$ be a graded commutative ring and $\cA$ a graded
commutative $\cK$-ring. Let $P$ and $Q$ be graded $\cA$-modules.
The graded $\cK$-module $\hm_\cK (P,Q)$ of graded $\cK$-module
homomorphisms $\Phi:P\to Q$ can be endowed with the two graded
$\cA$-module structures
\mar{ws11}\beq
(a\Phi)(p)= a\Phi(p),  \qquad  (\Phi\bll a)(p) = \Phi (a p),\qquad
a\in \cA, \quad p\in P, \label{ws11}
\eeq
called $\cA$- and $\cA^\bll$-module structures, respectively. Let
us put
\mar{ws12}\beq
\dl_a\Phi= a\Phi -(-1)^{[a][\Phi]}\Phi\bll a, \qquad a\in\cA.
\label{ws12}
\eeq
An element $\Delta\in\hm_\cK(P,Q)$ is said to be a $Q$-valued
graded differential operator of order $s$ on $P$ if
\be
\dl_{a_0}\circ\cdots\circ\dl_{a_s}\Delta=0
\ee
for any tuple of $s+1$ elements $a_0,\ldots,a_s$ of $\cA$. The set
$\dif_s(P,Q)$ of these operators inherits the graded module
structures (\ref{ws11}).

In particular, zero order graded differential operators obey the
condition
\be
\dl_a \Delta(p)=a\Delta(p)-(-1)^{[a][\Delta]}\Delta(ap)=0, \qquad
a\in\cA, \qquad p\in P,
\ee
i.e., they coincide with graded $\cA$-module morphisms $P\to Q$. A
first order graded differential operator $\Delta$ satisfies the
condition
\be
&& \dl_a\circ\dl_b\,\Delta(p)=
ab\Delta(p)- (-1)^{([b]+[\Delta])[a]}b\Delta(ap)-
(-1)^{[b][\Delta]}a\Delta(bp)+\\
&& \qquad (-1)^{[b][\Delta]+([\Delta]+[b])[a]}
=0, \qquad a,b\in\cA, \quad p\in P.
\ee

For instance, let $P=\cA$. Any zero order $Q$-valued graded
differential operator $\Delta$ on $\cA$ is defined by its value
$\Delta(\bb)$. Then there is a graded $\cA$-module isomorphism
\be
\dif_0(\cA,Q)=Q
\ee
via the association
\be
Q\ni q\mapsto \Delta_q\in \dif_0(\cA,Q),
\ee
where $\Delta_q$ is given by the equality $\Delta_q(\bb)=q$. A
first order $Q$-valued graded differential operator $\Delta$ on
$\cA$ fulfils the condition
\be
\Delta(ab)= \Delta(a)b+ (-1)^{[a][\Delta]}a\Delta(b)
-(-1)^{([b]+[a])[\Delta]} ab \Delta(\bb), \qquad  a,b\in\cA.
\ee
It is called a $Q$-valued graded derivation of $\cA$ if
$\Delta(\bb)=0$, i.e., the graded Leibniz rule
\mar{ws10}\beq
\Delta(ab) = \Delta(a)b + (-1)^{[a][\Delta]}a\Delta(b), \qquad
a,b\in \cA, \label{ws10}
\eeq
holds. One obtains at once that any first order graded
differential operator on $\cA$ falls into the sum
\be
\Delta(a)= \Delta(\bb)a +[\Delta(a)-\Delta(\bb)a]
\ee
of a zero order graded differential operator $\Delta(\bb)a$ and a
graded derivation $\Delta(a)-\Delta(\bb)a$. If $\dr$ is a graded
derivation of $\cA$, then $a\dr$ is so for any $a\in \cA$. Hence,
graded derivations of $\cA$ constitute a graded $\cA$-module
$\gd(\cA,Q)$, called the graded derivation module.

If $Q=\cA$, the graded derivation module $\gd\cA$ is also a Lie
superalgebra over the graded commutative ring $\cK$ with respect
to the superbracket
\mar{ws14}\beq
[u,u']=u\circ u' - (-1)^{[u][u']}u'\circ u, \qquad u,u'\in \cA.
\label{ws14}
\eeq
We have the graded $\cA$-module decomposition
\mar{ws15}\beq
\dif_1(\cA) = \cA \oplus\gd\cA. \label{ws15}
\eeq

Let us turn now to jets of graded modules. Given a graded
$\cA$-module $P$, let us consider the tensor product
$\cA\otimes_\cK P$ of graded $\cK$-modules $\cA$ and $P$. We put
\mar{ws16}\beq
\dl^b(a\otimes p)= (ba)\otimes p - (-1)^{[a][b]}a\otimes (bp),
\qquad p\in P, \quad a,b\in\cA.  \label{ws16}
\eeq
The $k$-order graded jet module $\cJ^k(P)$ of the module $P$ is
defined as the quotient of the graded $\cK$-module $\cA\otimes_\cK
P$ by its submodule generated by elements of type
\be
\dl^{b_0}\circ \cdots \circ\dl^{b_k}(a\otimes p).
\ee

In particular, the first order graded jet module $\cJ^1(P)$
consists of elements $a\ot_1 p$ modulo the relations
\mar{ws20}\beq
ab\otimes_1 p -(-1)^{[a][b]}b\otimes_1 (ap)
  -a\otimes_1(bp)  + \bb\ot_1(abp) =0. \label{ws20}
\eeq

For any $h\in\hm_\cA (\cA\ot P,Q)$, the equality
\be
\dl_b(h(a\ot p))=(-1)^{[h][b]}h(\dl^b(a\ot p))
\ee
holds. One then can show that any $Q$-valued graded differential
operator $\Delta$ of order $k$ on a graded $\cA$-module $P$
factorizes uniquely
\be
\Delta: P\ar^{J^k} \cJ^k(P)\ar Q
\ee
through the morphism
\be
J^k:p\ni p\mapsto \bb\ot_k p\in \cJ^k(P)
\ee
and some homomorphism ${\got f}^\Delta:\cJ^k(P)\to Q$.
Accordingly, the assignment $\Delta\mapsto {\got f}^\Delta$
defines an isomorphism
\mar{ws21}\beq
\dif_s(P,Q)=\hm_{\cA}(\cJ^s(P),Q). \label{ws21}
\eeq

Let us focus on the first order graded jet module $\cJ^1$ of $\cA$
consisting of the elements $a\otimes_1 b$, $a,b\in\cA$, subject to
the relations
\mar{ws22}\beq
ab\otimes_1 \bb -(-1)^{[a][b]}b\otimes_1 a
  -a\otimes_1b  + \bb\ot_1(ab) =0. \label{ws22}
\eeq
It is endowed with the $\cA$- and $\cA^\bll$-module structures
\be
c(a\ot_1 b)=(ca)\ot_1 b,\qquad c\bll(a\ot_1 b)= a\ot_1(cb).
\ee
There are canonical $\cA$- and $\cA^\bll$-module monomorphisms
\be
&& i_1: \cA \ni a  \mapsto a\otimes_1 \bb\in \cJ^1,\\
&& J^1: \cA\ni a\mapsto \bb\otimes_1 a\in \cJ^1,
\ee
such that $\cJ^1$, seen as a graded $\cA$-module, is generated by
the elements $J^1a$, $a\in \cA$. With these monomorphisms, we have
the canonical $\cA$-module splitting
\mar{ws23}\ben
&& \cJ^1=i_1(\cA)\oplus \cO^1, \label{ws23} \\
&&
aJ^1(b)= a\ot_1 b=ab\ot_1\bb + a(\bb\ot_1 b- b\ot_1\bb), \nonumber
\een
where the graded $\cA$-module $\cO^1$ is generated by the elements
\be
\bb\ot_1 b-b\ot_1 \bb, \qquad b\in\cA.
\ee
Let us consider the corresponding $\cA$-module epimorphism
\mar{ws30}\beq
h^1:\cJ^1\ni \bb\ot_1 b\mapsto \bb\ot_1 b-b\ot_1 \bb\in \cO^1
\label{ws30}
\eeq
and the composition
\mar{ws31}\beq
d=h^1\circ J_1: \cA \ni b \mapsto \bb\ot_1 b- b\ot_1\bb \in \cO^1.
\label{ws31}
\eeq
The equality
\be
d(ab)= a\otimes_1 b  +b\otimes_1 a- ab\otimes_1 \bb-ba\otimes_1
\bb  =(-1)^{[a][b]}bda + adb
\ee
shows that $d$ (\ref{ws31}) is an even $\cO^1$-valued derivation
of $\cA$. Seen as a graded $\cA$-module, $\cO^1$ is generated by
the elements $da$ for all $a\in\cA$.

In view of the splittings (\ref{ws15}) and (\ref{ws23}), the
isomorphism (\ref{ws21}) reduces to the isomorphism
\mar{ws47}\beq
\gd\cA=\cO^{1*}=\hm_{\cA}(\cO^1,\cA) \label{ws47}
\eeq
of $\gd\cA$ to the dual $\cO^{1*}$ of the graded $\cA$-module
$\cO^1$. It is given by the duality relations
\mar{ws41}\beq
\gd\cA\ni u\leftrightarrow
   \f_u\in \cO^{1*}, \qquad \f_u(da)=u(a), \qquad  a\in \cA.
   \label{ws41}
\eeq
Using this fact, let us construct a differential calculus over a
graded commutative $\cK$-ring $\cA$.

Let us consider the bigraded exterior algebra $\cO^*$ of a graded
module $\cO^1$. It consists of finite linear combinations of
monomials of the form
\mar{ws42}\beq
\f=a_0 da_1\w\cdots\w da_k, \qquad a_i\in \cA, \label{ws42}
\eeq
whose product obeys the juxtaposition rule
\be
(a_0d a_1)\w (b_0d b_1)=a_0d (a_1b_0)\w db_1- a_0a_1d b_0\w d b_1
\ee
and the bigraded commutative relations
\mar{ws45}\beq
\f\w \f'=(-1)^{|\f||\f'|+[\f][\f']}\f'\w\f. \label{ws45}
\eeq
In order to make $\cO^*$ to a differential algebra, let us define
the coboundary operator $d:\cO^1\to\cO^2$ by the rule
\be
d\f(u,u')=-u'(u(\f)) +(-1)^{[u][u']}u(u'(\f)) +[u',u](\f),
\ee
where $u,u'\in \gd\cA$, $\f\in\cO^1$, and $u$, $u'$ are both
graded derivatives of $\cA$ and $\cA$-valued forms on $\cO^1$. It
is readily observed that, by virtue of the relation (\ref{ws41}),
$(d\circ d)(a)=0$ for all $a\in \cA$. Then $d$ is extended to the
bigraded exterior algebra $\cO^*$ if its action on monomials
(\ref{ws42}) is defined as
\be
d(a_0 da_1\w\cdots\w da_k)=da_0\w da_1\w\cdots\w da_k.
\ee
This operator is nilpotent and fulfils the familiar relations
\mar{ws44}\beq
d(\f\w\f')= d\f\w\f' +(-1)^{|\f||\f'|}\f\w d\f'. \label{ws44}
\eeq
It makes $\cO^*$ into a differential bigraded algebra, called a
graded differential calculus over a graded commutative $\cK$-ring
$\cA$.

Furthermore, one can extend the duality relation (\ref{ws41}) to
the graded interior product of $u\in\gd\cA$ with any monomial $\f$
(\ref{ws42}) by the rules
\mar{ws46}\ben
&& u\rfloor(bda) =(-1)^{[u][b]}u(a), \nonumber\\
&& u\rfloor(\f\w\f')=
(u\rfloor\f)\w\f'+(-1)^{|\f|+[\f][u]}\f\w(u\rfloor\f').
\label{ws46}
\een
As a consequence, any graded derivation $u\in\gd\cA$ of $\cA$
yields a derivation
\mar{+117}\ben
&& \bL_u\f= u\rfloor d\f + d(u\rfloor\f), \qquad \f\in\cO^*, \qquad
u\in\gd\cA, \label{+117} \\
&& \bL_u(\f\w\f')=\bL_u(\f)\w\f' + (-1)^{[u][\f]}\f\w\bL_u(\f'), \nonumber
\een
of the bigraded algebra $\cO^*$ called the graded Lie derivative
of $\cO^*$.

\begin{remark}
Since $\gd\cA$ is a Lie $\cK$-superalgebra, let us consider the
Chevalley--Eilenberg complex $C^*[\gd\cA;\cA]$ where the graded
commutative ring $\cA$ is a regarded as a $\gd\cA$-module
\cite{fuks}. It is the complex
\mar{ws85}\beq
0\to \cA\ar^{\dl^0}C^1[\gd\cA;\cA]\ar^{\dl^1}\cdots
C^k[\gd\cA;\cA]\ar^{\dl^k}\cdots \label{ws85}
\eeq
where
\be
C^k[\gd\cA;\cA]=\hm_\cK(\op\w^k \gd\cA,\cA)
\ee
are $\gd\cA$-modules of $\cK$-linear graded morphisms of the
graded exterior products $\op\w^k \gd\cA$ of the $\cK$-module
$\gd\cA$ to $\cA$. Let us bring homogeneous elements of $\op\w^k
\gd\cA$ into the form
\be
\ve_1\w\cdots\ve_r\w\e_{r+1}\w\cdots\w \e_k, \qquad
\ve_i\in\gd\cA_0, \quad \e_j\in\gd\cA_1.
\ee
Then the coboundary operators of the complex (\ref{ws85}) are
given by the expression
\mar{ws86}\ben
&& \dl^{r+s-1}c(\ve_1\w\cdots\w\ve_r\w\e_1\w\cdots\w\e_s)=
\label{ws86}\\
&&\op\sum_{i=1}^r (-1)^{i-1}\ve_i
c(\ve_1\w\cdots\wh\ve_i\cdots\w\ve_r\w\e_1\w\cdots\e_s)+
\nonumber \\
&& \op\sum_{j=1}^s (-1)^r\ve_i
c(\ve_1\w\cdots\w\ve_r\w\e_1\w\cdots\wh\e_j\cdots\w\e_s)
+\nonumber\\
&& \op\sum_{1\leq i<j\leq r} (-1)^{i+j}
c([\ve_i,\ve_j]\w\ve_1\w\cdots\wh\ve_i\cdots\wh\ve_j
\cdots\w\ve_r\w\e_1\w\cdots\w\e_s)+\nonumber\\
&&\op\sum_{1\leq i<j\leq s} c([\e_i,\e_j]\w\ve_1\w\cdots\w
\ve_r\w\e_1\w\cdots
\wh\e_i\cdots\wh\e_j\cdots\w\e_s)+\nonumber\\
&& \op\sum_{1\leq i<r,1\leq j\leq s} (-1)^{i+r+1}
c([\ve_i,\e_j]\w\ve_1\w\cdots\wh\ve_i\cdots\w\ve_r\w
\e_1\w\cdots\wh\e_j\cdots\w\e_s).\nonumber
\een
The subcomplex $\cO^*[\gd\cA]$ of the complex (\ref{ws85}) of
$\cA$-linear morphisms is the graded Chevalley--Eilenberg
differential calculus over a graded commutative $\cK$-ring $\cA$.
Then one can show that the above mentioned graded differential
calculus $\cO^*$ is a subcomplex of the Chevalley--Eilenberg one
$\cO^*[\gd\cA]$.
\end{remark}

Following the construction of a connection in commutative geometry
\cite{book09,book00,sard09}, one comes to the notion of a
connection on modules over a graded commutative $\Bbb R$-ring
$\cA$.

\begin{definition} \label{ws33} \mar{ws33}
A connection on a graded $\cA$-module $P$ is an $\cA$-module
morphism
\mar{ws34}\beq
\gd\cA\ni u\mapsto \nabla_u\in \dif_1(P,P) \label{ws34}
\eeq
such that the first order differential operators $\nabla_u$ obey
the Leibniz rule
\mar{ws35}\beq
\nabla_u (ap)= u(a)p+ (-1)^{[a][u]}a\nabla_u(p), \quad a\in \cA,
\quad p\in P. \label{ws35}
\eeq
\end{definition}

\begin{definition} \label{ws36} \mar{ws36}
Let $P$ in Definition \ref{ws33} be a graded commutative
$\cA$-ring and $\gd P$ the derivation module of $P$ as a graded
commutative $\cK$-ring. A connection on a graded commutative
$\cA$-ring $P$ is a $\cA$-module morphism
\mar{ws37}\beq
\gd\cA\ni u\mapsto \nabla_u\in \gd P, \label{ws37}
\eeq
which is a connection on $P$ as an $\cA$-module, i.e., obeys the
Leibniz rule (\ref{ws35}).
\end{definition}

\section{Geometry of graded manifolds}

By a graded manifold of dimension $(n,m)$ is meant a local-ringed
space $(Z,\gA)$ where $Z$ is an $n$-dimensional smooth manifold
$Z$ and $\gA=\gA_0\oplus\gA_1$ is a sheaf of graded commutative
algebras of rank $m$ such that \cite{bart}:

(i) there is the exact sequence of sheaves
\mar{cmp140}\beq
0\to \cR \to\gA \op\to^\si C^\infty_Z\to 0, \qquad
\cR=\gA_1+(\gA_1)^2,\label{cmp140}
\eeq
where $C^\infty_Z$ is the sheaf of smooth real functions on $Z$;

(ii) $\cR/\cR^2$ is a locally free sheaf of $C^\infty_Z$-modules
of finite rank (with respect to pointwise operations), and the
sheaf $\gA$ is locally isomorphic to the the exterior product
$\w_{C^\infty_Z}(\cR/\cR^2)$.

The sheaf $\gA$ is called a structure sheaf of the graded manifold
$(Z,\gA)$, while the manifold $Z$ is said to be a body of
$(Z,\gA)$. Sections of the sheaf $\gA$ are called graded
functions. They make up a graded commutative $C^\infty(Z)$-ring
$\gA(Z)$.

A graded manifold $(Z,\gA)$ has the following local structure.
Given a point $z\in Z$, there exists its open neighborhood $U$,
called a splitting domain, such that
\mar{+54}\beq
\gA(U)\cong C^\infty(U)\ot\w\Bbb R^m. \label{+54}
\eeq
It means that the restriction $\gA|_U$ of the structure sheaf
$\gA$ to $U$ is isomorphic to the sheaf $C^\infty_U\ot\w\Bbb R^m$
of sections of some exterior bundle
\be
\w E^*_U= U\times \w\Bbb R^m\to U.
\ee

The well-known Batchelor's theorem \cite{bart} states that such a
structure of graded manifolds is global.

\begin{theorem} \label{lmp1} \mar{lmp1}
Let $(Z,\gA)$ be a graded manifold. There exists a vector bundle
$E\to Z$ with an $m$-dimensional typical fibre $V$ such that the
structure sheaf $\gA$ of $(Z,\gA)$ is isomorphic to the structure
sheaf $\gA_E$ of sections of the exterior bundle $\w E^*$, whose
typical fibre is the Grassmann algebra $\w V^*$.
\end{theorem}

\begin{proof}
The local sheaves $C^\infty_U\ot\w\Bbb R^m$ are glued into the
global structure sheaf $\gA$ of the graded manifold $(Z,\gA)$ by
means of transition functions, which are assembled into a cocycle
of the sheaf $\Is(\w\Bbb R^m)^\infty$ of smooth mappings from $Z$
to $\Is(\w\Bbb R^m)$. The proof is based on the bijection between
the cohomology sets $H^1(Z;\Is(\w\Bbb R^m)^\infty)$ and
$H^1(Z;GL(m,\Bbb R^m)^\infty)$.
\end{proof}

It should be emphasized that Batchelor's isomorphism in Theorem
\ref{lmp1} fails to be canonical. At the same time, there are many
physical models where a vector bundle $E$ is introduced from the
beginning. In this case, it suffices to consider the structure
sheaf $\gA_E$ of the exterior bundle $\w E^*$. We agree to call
the pair $(Z,\gA_E)$ a simple graded manifold. Its automorphisms
are restricted to those induced by automorphisms of the vector
bundle $E\to Z$, called the characteristic vector bundle of the
simple graded manifold $(Z,\gA_E)$. Accordingly, the structure
module
\be
\gA_E(Z)=\w E^*(Z)
\ee
of the sheaf $\gA_E$ (and of the exterior bundle $\w E^*$) is said
to be the structure module of the simple graded manifold
$(Z,\gA_E)$.

Combining Batchelor Theorem \ref{lmp1} and classical Serre--Swan
theorem \cite{book09,sard09}, we come to the following Serre--Swan
theorem for graded manifolds.

\begin{theorem} \label{vv0} \mar{vv0}
Let $Z$ be a smooth manifold. A graded commutative
$C^\infty(Z)$-algebra $\cA$ is isomorphic to the structure ring of
a graded manifold with a body $Z$ iff it is the exterior algebra
of some projective $C^\infty(Z)$-module of finite rank.
\end{theorem}

Given a simple graded manifold $(Z,\gA_E)$, every trivialization
chart $(U; z^A,y^a)$ of the vector bundle $E\to Z$ is a splitting
domain of $(Z,\gA_E)$. Graded functions on such a chart are
$\La$-valued functions
\mar{z785}\beq
f=\op\sum_{k=0}^m \frac1{k!}f_{a_1\ldots a_k}(z)c^{a_1}\cdots
c^{a_k}, \label{z785}
\eeq
where $f_{a_1\cdots a_k}(z)$ are smooth functions on $U$ and
$\{c^a\}$ is the  fibre basis for $E^*$. In particular, the sheaf
epimorphism $\si$ in (\ref{cmp140}) is induced by the body map of
$\La$. We agree to call $\{z^A,c^a\}$ the local basis for the
graded manifold $(Z,\gA_E)$. Transition functions
\be
y'^a=\rho^a_b(z^A)y^b
\ee
of bundle coordinates on $E\to Z$ induce the corresponding
transformation
\mar{+6}\beq
c'^a=\rho^a_b(z^A)c^b \label{+6}
\eeq
of the associated local basis for the graded manifold $(Z,\gA_E)$
and the according coordinate transformation law of graded
functions (\ref{z785}).

\begin{remark}
Although graded functions are locally represented by $\La$-valued
functions (\ref{z785}), they are not $\La$-valued functions on a
manifold $Z$ because of the transformation law (\ref{+6}).
\end{remark}

\begin{remark}
Let us note that general automorphisms of a graded manifold take
the form
\mar{+95}\beq
c'^a=\rho^a(z^A,c^b), \label{+95}
\eeq
where $\rho^a(z^A,c^b)$ are local graded functions. Considering
simple graded manifolds, we actually restrict the class of graded
manifold transformations (\ref{+95}) to the linear ones
(\ref{+6}), compatible with given Batchelor's isomorphism.
\end{remark}

Let $E\to Z$ and $E'\to Z$ be vector bundles and $\Phi: E\to E'$
their bundle morphism over a morphism $\zeta: Z\to Z'$. Then every
section $s^*$ of the dual bundle $E'^*\to Z'$ defines the
pull-back section $\Phi^*s^*$ of the dual bundle $E^*\to Z$ by the
law
\be
v_z\rfloor \Phi^*s^*(z)=\Phi(v_z)\rfloor s^*(\zeta(z)), \qquad
v_z\in E_z.
\ee
 It follows that a linear bundle morphism
$\Phi$ yields a morphism
\mar{w901}\beq
S\Phi: (Z,\gA_E) \to (Z',\gA_{E'}) \label{w901}
\eeq
 of simple graded manifolds seen as
local-ringed spaces \cite{bart,book00,sard09}. This is the pair
$(\zeta,\zeta_*\circ\Phi^*)$ of the morphism $\zeta$ of the body
manifolds and the composition of the pull-back
\be
 \gA_{E'}\ni
f\mapsto \Phi^*f\in\gA_E
\ee
 of graded functions and the direct
image $\zeta_*$ of the sheaf $\gA_E$ onto $Z'$. Relative to local
bases $(z^A,c^a)$ and $(z'^A,c'^a)$ for $(Z,\gA_E)$ and
$(Z',\gA_{E'})$ respectively, the morphism (\ref{w901}) reads
\be
S\Phi(z)=\zeta(z), \qquad S\Phi(c'^a)=\Phi^a_b(z)c^b.
\ee
Accordingly, the pull-back onto $Z$ of graded exterior forms on
$Z'$ is defined.

Given a graded manifold $(Z,\gA)$, by the sheaf $\gd\gA$ of graded
derivations of $\gA$ is meant a subsheaf of endomorphisms of the
structure sheaf $\gA$ such that any section $u$ of $\gd\gA$ over
an open subset $U\subset Z$ is a graded derivation of the graded
algebra $\gA(U)$. Conversely, one can show that, given open sets
$U'\subset U$, there is a surjection of the derivation modules
\be
\gd(\gA(U))\to \gd(\gA(U'))
\ee
\cite{bart}. It follows that any graded derivation of the local
graded algebra $\gA(U)$ is also a local section over $U$ of the
sheaf $\gd\gA$. Sections of $\gd\gA$ are  called graded vector
fields on the graded manifold $(Z,\gA)$. They make up the graded
derivation module $\gd\gA(Z)$ of the graded commutative $\Bbb
R$-ring $\gA(Z)$. This module is a Lie superalgebra with respect
to the superbracket (\ref{ws14}).

In comparison with general theory of graded manifolds, an
essential simplification is that graded vector fields on a simple
graded manifold $(Z,\gA_E)$ can be seen as sections of a vector
bundle as follows.

Due to the vertical splitting
\be
VE\cong E\times E,
\ee
the vertical tangent bundle $VE$ of $E\to Z$ can be provided with
the fibre bases $\{\dr/\dr c^a\}$, which are the duals of the
bases $\{c^a\}$. These are the fibre bases for
\be
\pr_2VE\cong E.
\ee
Then graded vector fields on a trivialization chart $(U;z^A,y^a)$
of $E$ read
\mar{hn14}\beq
u= u^A\dr_A + u^a\frac{\dr}{\dr c^a}, \label{hn14}
\eeq
where $u^\la, u^a$ are local graded functions on $U$. In
particular,
\be
\frac{\dr}{\dr c^a}\circ\frac{\dr}{\dr c^b} =-\frac{\dr}{\dr
c^b}\circ\frac{\dr}{\dr c^a}, \qquad \dr_A\circ\frac{\dr}{\dr
c^a}=\frac{\dr}{\dr c^a}\circ \dr_A.
\ee
The derivations (\ref{hn14}) act on graded functions
$f\in\gA_E(U)$ (\ref{z785}) by the rule
\mar{cmp50a}\beq
u(f_{a\ldots b}c^a\cdots c^b)=u^A\dr_A(f_{a\ldots b})c^a\cdots c^b
+u^k f_{a\ldots b}\frac{\dr}{\dr c^k}\rfloor (c^a\cdots c^b).
\label{cmp50a}
\eeq
This rule implies the corresponding coordinate transformation law
\be
u'^A =u^A, \qquad u'^a=\rho^a_ju^j +u^A\dr_A(\rho^a_j)c^j
\ee
of graded vector fields. It follows that graded vector fields
(\ref{hn14}) can be represented by sections of the vector bundle
$\cV_E\to Z$ which is locally isomorphic to the vector bundle
\mar{+243}\beq
\cV_E|_U\approx\w E^*\op\ot_Z(E\op\oplus_Z TZ)|_U, \label{+243}
\eeq
and is characterized by an atlas of bundle coordinates
\be
(z^A,z^A_{a_1\ldots a_k},v^i_{b_1\ldots b_k}), \qquad
k=0,\ldots,m,
\ee
possessing the transition functions
\be
 && z'^A_{i_1\ldots
i_k}=\rho^{-1}{}_{i_1}^{a_1}\cdots
\rho^{-1}{}_{i_k}^{a_k} z^A_{a_1\ldots a_k}, \\
&& v'^i_{j_1\ldots j_k}=\rho^{-1}{}_{j_1}^{b_1}\cdots
\rho^{-1}{}_{j_k}^{b_k}\left[\rho^i_jv^j_{b_1\ldots b_k}+
\frac{k!}{(k-1)!} z^A_{b_1\ldots b_{k-1}}\dr_A\rho^i_{b_k}\right],
\ee
which fulfil the cocycle condition. Thus, the derivation module
$\gd\gA_E(Z)$ is isomorphic to the structure module $\cV_E(Z)$ of
global sections of the vector bundle $\cV_E\to Z$.

There is the exact sequence
\mar{1030}\beq
0\to \w E^*\op\ot_Z E\to\cV_E\to \w E^*\op\ot_Z TZ\to 0
\label{1030}
\eeq
of vector bundles over $Z$. Its splitting
\mar{cmp70}\beq
\wt\g:\dot z^A\dr_A \mapsto \dot z^A(\dr_A
+\wt\g_A^a\frac{\dr}{\dr c^a}) \label{cmp70}
\eeq
transforms every vector field $\tau$ on $Z$ into the graded vector
field
\mar{ijmp10}\beq
\tau=\tau^A\dr_A\mapsto \nabla_\tau=\tau^A(\dr_A
+\wt\g_A^a\frac{\dr}{\dr c^a}), \label{ijmp10}
\eeq
which is a graded derivation of the graded commutative $\Bbb
R$-ring $\gA_E(Z)$ satisfying the Leibniz rule
\be
\nabla_\tau(sf)=(\tau\rfloor ds)f +s\nabla_\tau(f), \quad
f\in\gA_E(Z), \quad s\in C^\infty(Z).
\ee
It follows that the splitting (\ref{cmp70}) of the exact sequence
(\ref{1030}) yields a connection on the graded commutative
$C^\infty(Z)$-ring $\gA_E(Z)$ in accordance with Definition
\ref{ws36}. It is called a graded connection on the simple graded
manifold $(Z,\gA_E)$. In particular, this connection provides the
corresponding horizontal splitting
\be
u= u^A\dr_A + u^a\frac{\dr}{\dr c^a}=u^A(\dr_A
+\wt\g_A^a\frac{\dr}{\dr c^a}) + (u^a- u^A\wt\g_A^a)\frac{\dr}{\dr
c^a}
\ee
of graded vector fields.

\begin{remark} \label{+94} \mar{+94}
By virtue of the isomorphism (\ref{+54}), any connection $\wt \g$
on a graded manifold $(Z,\gA)$, restricted to a splitting domain
$U$, takes the form (\ref{cmp70}). Given two splitting domains $U$
and $U'$ of $(Z,\gA)$ with the transition functions (\ref{+95}),
the connection components $\wt\g^a_A$ obey the transformation law
\mar{+96}\beq
\wt\g'^a_A= \wt\g^b_A\frac{\dr}{\dr c^b}\rho^a +\dr_A\rho^a.
\label{+96}
\eeq
If $U$ and $U'$ are the trivialization charts of the same vector
bundle $E$ in Theorem \ref{lmp1} together with the transition
functions (\ref{+6}), the transformation law (\ref{+96}) takes the
form
\mar{+97}\beq
\wt\g'^a_A= \rho^a_b(z)\wt\g^b_A +\dr_A\rho^a_b(z)c^b. \label{+97}
\eeq
\end{remark}

\begin{remark}
It should be emphasized that the above notion of a graded
connection is a connection on the graded commutative ring
$\gA_E(Z)$ seen as a $C^\infty(Z)$-module. It differs from that of
a  connection on a graded fibre bundle $(Z,\gA)\to (X,\cB)$ in
\cite{alm}. The latter is a connection on a graded $\cB(X)$-module
represented by a section of the jet graded bundle $J^1(Z/X)\to
(Z,\gA)$ of sections of the graded fibre bundle
   $(Z,\gA)\to (X,\cB)$ \cite{rup}.
\end{remark}

\begin{example} \label{1031} \mar{1031}
Every linear connection
\be
\g=dz^A\ot (\dr_A +\g_A{}^a{}_by^b \dr_a)
\ee
on the vector bundle $E\to Z$ yields the graded connection
\mar{cmp73}\beq
\g_S=dz^A\ot (\dr_A +\g_A{}^a{}_bc^b\frac{\dr}{\dr c^a})
\label{cmp73}
\eeq
on the simple graded manifold $(Z,\gA_E)$. In view of Remark
\ref{+94}, $\g_S$ is also a graded connection on the graded
manifold
\be
(Z,\gA)\cong (Z,\gA_E),
\ee
but its linear form (\ref{cmp73}) is not maintained under the
transformation law (\ref{+96}).
\end{example}

The curvature of the graded connection $\nabla_\tau$
(\ref{ijmp10}) is defined by the familiar expression
\mar{+110}\ben
&&
R(\tau,\tau')=[\nabla_\tau,\nabla_{\tau'}]-\nabla_{[\tau,\tau']},\nonumber\\
&& R(\tau,\tau') =\tau^A\tau'^B R^a_{AB}\frac{\dr}{\dr c^a}:
\gA_E(Z)\to \gA_E(Z),
\nonumber\\
&&R^a_{AB} =\dr_A\wt\g^a_B-\dr_B\wt\g^a_A
+\wt\g^k_A\frac{\dr}{\dr c^k}\wt\g^a_B -
   \wt\g^k_B\frac{\dr}{\dr c^k} \wt\g^a_A.\label{+110}
\een
It can also be written in the form
\mar{+111}\ben
&&R:\gA_E(Z)\to \cO^2(Z)\ot \gA_E(Z), \nonumber \\
&& R =\frac12 R_{AB}^a dz^A\w dz^B\ot\frac{\dr}{\dr c^a}. \label{+111}
\een

Let now $\cV^*_E\to  Z$ be a vector bundle which is the pointwise
$\w E^*$-dual of the vector bundle $\cV_E\to Z$. It is locally
isomorphic to the vector bundle
\mar{+244}\beq
\cV^*_E|_U\approx \w E^*\op\ot_Z(E^*\op\oplus_Z T^*Z)|_U.
\label{+244}
\eeq
With respect to the dual bases $\{dz^A\}$ for $T^*Z$ and
$\{dc^b\}$ for
\be
\pr_2V^*E\cong E^*,
\ee
sections of the vector bundle $\cV^*_E$ take the coordinate form
\be
\f=\f_A dz^A + \f_adc^a,
\ee
together with transition functions
\be
\f'_a=\rho^{-1}{}_a^b\f_b, \qquad \f'_A=\f_A
+\rho^{-1}{}_a^b\dr_A(\rho^a_j)\f_bc^j.
\ee
They are regarded as graded exterior one-forms on the graded
manifold $(Z,\gA_E)$, and make up the $\gA_E(Z)$-dual
\be
\cC^1_E=\gd\gA_E(Z)^*
\ee
of the derivation module
\be
\gd\gA_E(Z)=\cV_E(Z).
\ee
Conversely,
\be
\gd\gA_E(Z)=(\cC^1_E)^*.
\ee
The duality morphism is given by the graded interior product
\mar{cmp65}\beq
u\rfloor \f=u^A\f_A + (-1)^{\nw{\f_a}}u^a\f_a. \label{cmp65}
\eeq
In particular, the dual of the exact sequence (\ref{1030}) is the
exact sequence
\mar{cmp72}\beq
0\to \w E^*\op\ot_ZT^*Z\to\cV^*_E\to \w E^*\op\ot_Z E^*\to 0.
\label{cmp72}
\eeq
Any graded connection $\wt\g$ (\ref{cmp70}) yields the splitting
of the exact sequence (\ref{cmp72}), and determines the
corresponding decomposition of graded one-forms
\be
\f=\f_A dz^A + \f_adc^a =(\f_A+\f_a\wt\g_A^a)dz^A +\f_a(dc^a
-\wt\g_A^adz^A).
\ee

Higher degree graded exterior forms  are defined as sections of
the exterior bundle $\op\w^k_Z\cV^*_E$. They make up a bigraded
algebra $\cC^*_E$ which is isomorphic to the bigraded exterior
algebra of the graded module $\cC^1_E$ over $\cC^0_E=\gA(Z)$. This
algebra is locally generated by graded forms $dz^A$, $dc^i$ such
that
\mar{+113'}\beq
dz^A\w dc^i=-dc^i\w dz^A, \qquad dc^i\w dc^j= dc^j\w dc^i.
\label{+113'}
\eeq

The graded exterior differential $d$ of graded functions is
introduced by the condition $u\rfloor df=u(f)$ for an arbitrary
graded vector field $u$, and is extended uniquely to graded
exterior forms by the rule (\ref{ws44}). It is given by the
coordinate expression
\be
d\f= dz^A \w \dr_A\f +dc^a\w \frac{\dr}{\dr c^a}\f,
\ee
where the derivatives $\dr_\la$, $\dr/\dr c^a$ act on coefficients
of graded exterior forms by the formula (\ref{cmp50a}), and they
are graded commutative with the graded forms $dz^A$ and $dc^a$.
The formulae (\ref{ws45}) -- (\ref{+117}) hold.

The graded exterior differential $d$ makes $\cC^*_E$ into a
bigraded differential algebra whose de Rham complex reads
\mar{+137}\beq
0\to\Bbb R\to \gA_E(Z)\ar^d \cC^1_E \ar^d \cdots \cC^k_E \ar^d
\cdots. \label{+137}
\eeq
Its cohomology $H^*_{GR}(Z)$  is called the graded de Rham
cohomology of the graded manifold $(Z,\gA_E)$. One can compute
this cohomology with the aid of the abstract de Rham theorem. Let
$\gO^k\gA_E$ denote the sheaf of germs of graded $k$-forms on
$(Z,\gA_E)$. Its structure module is $\cC^k_E$. These sheaves make
up the complex
\mar{1033}\beq
0\to\Bbb R\ar \gA_E \ar^d \gO^1\gA_E\ar^d\cdots
\gO^k\gA_E\ar^d\cdots. \label{1033}
\eeq
Its members $\gO^k\gA_E$ are sheaves of $C^\infty_Z$-modules on
$Z$ and, consequently, are fine and acyclic. Furthermore, the
Poincar\'e lemma for graded exterior forms holds \cite{bart}. It
follows that the complex (\ref{1033}) is a fine resolution of the
constant sheaf $\Bbb R$ on the manifold $Z$.  Then, by virtue of
the abstract de Rham theorem \cite{book09,sard09}, there is
   an isomorphism
\mar{+136}\beq
H^*_{GR}(Z)=H^*(Z;\Bbb R)=H^*(Z) \label{+136}
\eeq
of the graded de Rham cohomology $H^*_{GR}(Z)$ to the de Rham
cohomology $H^*(Z)$ of the smooth manifold $Z$. Moreover, the
cohomology isomorphism (\ref{+136}) accompanies the cochain
monomorphism $\cO^*(Z)\to \cC^*_E$ of the de Rham complex
$\cO^*(Z)$ of smooth exterior forms on $Z$ to the graded de Rham
complex (\ref{+137}). Hence, any closed graded exterior form is
split into a sum $\f=d\si +\vf$ of an exact graded exterior form
$d\si\in \cO^*\gA_E$ and a closed exterior form $\vf\in \cO^*(Z)$
on $Z$.

\section{Superfunctions}

By analogy with smooth manifolds, supermanifolds are constructed
by gluing together of open subsets of supervector spaces $B^{n,m}$
with the aid  of transition superfunctions \cite{bart,book00}.
Therefore, let us start with the notion of a superfunction.

Though there are different classes of superfunctions, they can be
introduced in the same manner as follows.

Let
\be
B^{n,m}=\La^n_0\oplus \La^m_1,\qquad n,m\geq 0,
\ee
be a supervector space, where $\La$ is a Grassmann algebra of rank
$0<N\geq m$. Let
\be
\si^{n,m}: B^{n,m}\to \Bbb R^n, \qquad s: B^{n,m}\to R^{n,m}=R^n_0
\oplus R_1^m
\ee
be the corresponding body and soul maps (see the decomposition
(\ref{+11})). Then any element $q\in B^{n,m}$ is uniquely split as
\mar{+10}\beq
q=(x,y)=(\si(x^i) + s(x^i))e^0_i + y^je^1_j, \label{+10}
\eeq
where $\{e^0_i,e^1_j\}$ is a basis for $B^{n,m}$ and $\si(x^i)\in
\Bbb R$, $s(x^i)\in R_0$, $y^j\in R_1$.

Let $\La'$ be another Grassmann algebra of rank $0\leq N'\leq N$
which is treated as a subalgebra of $\La$, i.e., the basis
$\{c^a\}$, $a=1,\ldots,N'$, for $\La'$ is a subset of the basis
$\{c^i\}$, $i=1,\ldots,N$, for $\La$. Given an open subset
$U\subset \Bbb R^n$, let us consider a $\La'$-valued function
\mar{+12}\beq
f(z)=\op\sum_{k=0}^{N'} \frac1{k!}f_{a_1\ldots
a_k}(z)c^{a_1}\cdots c^{a_k} \label{+12}
\eeq
on $U$ with smooth coefficients $f_{a_1\cdots a_k}(z)$, $z\in U$.
It is a graded function on $U$. Its prolongation to
$(\si^{n,0})^{-1}(U)\subset B^{n,0}$ is defined as the formal
Taylor series
\mar{+14}\beq
f(x)= \op\sum_{k=0}^{N'} \frac1{k!}\left[
   \op\sum_{p=0}^N\frac{1}{p!}\frac{\dr^pf_{a_1\ldots a_k}}{\dr
z^{i_1}\cdots \dr z^{i_p}}(\si(x))s(x^{i_1})\cdots
s(x^{i_p})\right]c^{a_1}\cdots c^{a_k}. \label{+14}
\eeq
Then a superfunction $F(q)$ on
\be
(\si^{n,m})^{-1}(U)\subset B^{n,m}
\ee
is given by a sum
\mar{+13}\beq
F(q)=F(x,y)= \op\sum_{r=0}^m \frac1{r!} f_{j_1\ldots
j_r}(x)y^{j_1}\cdots y^{j_r}, \label{+13}
\eeq
where $f_{j_1\ldots j_r}(x)$ are functions (\ref{+14}). However,
the representation of a superfunction $F(x,y)$ by the sum
(\ref{+13}) need not be unique.

The germs of superfunctions (\ref{+13}) constitute the sheaf
$\gS_{N'}$ of graded commutative $\La'$-algebras on $B^{n,m}$, but
it is not a sheaf of $C^\infty_{B^{n,m}}$-modules since
superfunctions are expressed in Taylor series.

Using the representation (\ref{+13}), one can define derivatives
of superfunctions as follows. Let $f(x)$ be a superfunction  on
$B^{n,0}$. Since $f$, by definition, is the Taylor series
   (\ref{+14}), its
partial derivative along an even coordinate $x^i$ is defined in a
natural way as
\mar{+16}\ben
&& \dr_if(x)=(\dr_if)(\si(x),s(x))=
\label{+16}\\ && \qquad \op\sum_{k=0}^{N'} \frac1{k!}\left[
   \op\sum_{p=0}^N\frac{1}{p!}\frac{\dr^{p+1}f_{a_1\ldots a_k}}{\dr
z^i\dr z^{i_1}\cdots \dr z^{i_p}}(\si(x))s(x^{i_1})\cdots
s(x^{i_p})\right]c^{a_1}\cdots c^{a_k}. \nonumber
\een
This even derivative is extended to superfunctions $F$ on
$B^{n,m}$ in spite of the fact that the representation (\ref{+13})
is not necessarily unique. However, the definition of odd
derivatives of superfunctions is more intricate.

Let $\gS_{N'}^0\subset \gS_{N'}$ be the subsheaf of superfunctions
$F(x,y)=f(x)$ (\ref{+14}) independent of the odd arguments $y^j$.
Let $\w\Bbb R^m$ be a Grassmann algebra generated by
$(a^1,\ldots,a^m)$. The expression (\ref{+13}) implies that, for
any open subset $U\subset B^{n,m}$, there exists the sheaf
morphism
\mar{+15,7}\ben
&&\la: \gS^0_{N'}\ot\w\Bbb R^m \to \gS_{N'}, \label{+15}\\
&& \la(x,y):\op\sum_{r=0}^m \frac1{r!} f_{j_1\ldots
j_r}(x)\ot(a^{j_1}\cdots a^{j_r})\to \label{+17}\\
&& \qquad \op\sum_{r=0}^m \frac1{r!} f_{j_1\ldots
j_r}(x)y^{j_1}\cdots y^{j_r},\nonumber
\een
over $B^{n,m}$. Clearly, the morphism $\la$ (\ref{+15}) is an
epimorphism. One can show that this epimorphism is injective and,
consequently, is an isomorphism if and only if
\mar{+20}\beq
N-N'\geq m \label{+20}
\eeq
\cite{bart}. Roughly speaking, in this case, there exists a tuple
of elements $y^{j_1},\ldots,y^{j_r}\in \La$ for each superfunction
$f$ such that
\be
\la(f\ot(a^{j_1}\cdots a^{j_r}))\neq 0
\ee
at the point $(x,y^{j_1},\ldots,y^{j_m})$ of $B^{n,m}$.

If the condition (\ref{+20}) holds, the representation of each
superfunction $F(x,y)$ by the sum (\ref{+13}) is unique, and it is
an image of some section $f\ot a$ of the sheaf
$\gS^0_{N'}\ot\w\Bbb R^m$ with respect to the morphism $\la$
(\ref{+17}). Then an odd derivative of $F$ is defined as
\be
\frac{\dr}{\dr y^j}(\la(f\ot y))=\la (f\ot \frac{\dr}{\dr
a^j}(a)).
\ee
This definition is consistent only if $\la$ is an isomorphism,
i.e., the relation (\ref{+20}) holds. If otherwise, there exists a
non-vanishing element $f\ot a$ such that
\be
\la(f\ot a)=0,
\ee
whereas
\be
\la (f\ot \dr_j(a))\neq 0.
\ee
For instance, if
\be
N-N'=m-1,
\ee
such an element is
\be
f\ot a=c^1\cdots c^{N'}\ot(a^1\cdots a^m).
\ee

In order to classify superfunctions, we follow the terminology of
\cite{bart,rog86,rog07}.

$\bullet$ If $N'=N$, one deals with $G^\infty$-superfunctions,
introduced in \cite{rog80}. In this case, the inequality
(\ref{+20}) is not satisfied, unless $m=0$.

$\bullet$ If the condition (\ref{+20}) holds,
$\gS_{N'}=\ccG\cH_{N'}$ is the sheaf of
$GH^\infty$-superfunctions.

$\bullet$ In particular, if $N'=0$, the condition (\ref{+20}) is
satisfied, and $\gS_{N'}=\cH^\infty$ is the sheaf of
$H^\infty$-superfunctions
\mar{+41}\beq
F(x,y)= \op\sum_{r=0}^m \frac1{r!}\left[
   \op\sum_{p=0}^N\frac{1}{p!}\frac{\dr^pf_{j_1\ldots j_r}}{\dr
z^{i_1}\cdots \dr z^{i_p}}(\si(x))s(x^{i_1})\cdots
s(x^{i_p})\right]y^{j_1}\cdots y^{j_r}, \label{+41}
\eeq
where $f_{j_1\ldots j_r}$ are real functions \cite{batch2,dewt}.

Superfunctions of the above three types are called smooth
superfunctions. The fourth type of superfunctions is the
following.

Given the sheaf $\ccG\cH_{N'}$ of $GH^\infty$-superfunctions on a
supervector space $B^{n,m}$, let us define the sheaf of graded
commutative $\La$-algebras
\mar{+21}\beq
\ccG_{N'}=\ccG\cH_{N'}\op\ot_{\La'} \La, \label{+21}
\eeq
where $\La$ is regarded as a graded commutative $\La'$-algebra.
The sheaf $\ccG_{N'}$ (\ref{+21}) possesses the following
important properties \cite{bart}.

$\bullet$ There is the evaluation morphism
\mar{+22}\beq
\dl:\ccG_{N'}\ni F\ot a\mapsto Fa \in C^\La_{B^{n,m}}, \label{+22}\\
\eeq
where
\be
C^\La_{B^{n,m}}\cong C^0_{B^{n,m}}\ot\La
\ee
is the sheaf of continuous $\La$-valued functions on $B^{n,m}$.
Its image is isomorphic to the sheaf $\ccG^\infty$ of
$G^\infty$-superfunctions on $B^{n,m}$.

$\bullet$ For any two integers $N'$ and $N''$ satisfying the
condition (\ref{+20}), there exists the canonical isomorphism
between the sheaves $\ccG_{N'}$ and $\ccG_{N''}$. Therefore, one
can define the canonical sheaf $\ccG_{n,m}$ of graded commutative
$\La$-algebras on the supervector space $B^{n,m}$ whose sections
can be seen as tensor products $F\ot a$ of
$H^\infty$-superfunctions $F$ (\ref{+41}) and elements $a\in\La$.
They are called $G$-superfunctions.

$\bullet$ The sheaf $\gd\ccG_{n,m}$ of graded derivations of the
sheaf $\ccG_{n,m}$ is a locally free sheaf of $\ccG_{n,m}$-modules
of rank $(n,m)$. On any open set $U\subset B^{n,m}$, the
$\ccG_{n,m}(U)$-module $\gd\ccG_{n,m}(U)$ is generated by the
derivations $\dr/\dr x^i$, $\dr/\dr y^j$ which act on
$\ccG_{n,m}(U)$ by the rule
\mar{+83}\beq
\frac{\dr}{\dr x^i}(F\ot a)=\frac{\dr F}{\dr x^i}\ot a, \qquad
\frac{\dr}{\dr y^j}(F\ot a)=\frac{\dr F}{\dr y^j}\ot a.
\label{+83}
\eeq

These properties of $G$-superfunctions make $G$-supermanifolds
most suitable for differential geometric constructions.

\section{Supermanifolds}

A paracompact topological space $M$ is said to be an
$(n,m)$-dimensional smooth supermanifold if it admits an atlas
\be
\Psi=\{U_\zeta,\f_\zeta\}, \qquad \f_\zeta: U_\zeta\to B^{n,m},
\ee
such that the transition functions $\f_\zeta\circ\f_\xi^{-1}$ are
supersmooth. Obviously, a smooth supermanifold of dimension
$(n,m)$ is also a real smooth manifold of dimension
$2^{N-1}(n+m)$. If transition superfunctions are $H^\infty$-,
$G^\infty$- or $GH^\infty$-superfunctions, one deals with
$H^\infty$-, $G^\infty$- or $GH^\infty$-supermanifolds,
respectively. This definition is equivalent to the following one.

\begin{definition} \label{+27} \mar{+27}
A smooth supermanifold is a graded local-ringed space $(M,\gS)$
which is locally isomorphic to $(B^{n,m},\cS)$, where $\cS$ is one
of the sheaves of smooth superfunctions on $B^{n,m}$. The sheaf
$\cS$ is called the structure sheaf of a smooth supermanifold.
\end{definition}

In accordance with Definition \ref{+27}, by a morphism of smooth
supermanifolds is meant their morphism $(\vf,\Phi)$ as graded
local-ringed spaces, where $\Phi$ is an even graded morphism. In
particular, every morphism $\vf: M\to M'$ yields the smooth
supermanifold morphism $(\vf,\Phi=\vf^*)$.

Smooth supermanifolds however are effected by serious
inconsistencies as follows. Since odd derivatives of
$G^\infty$-superfunctions are ill defined, the sheaf of
derivations of the sheaf of $G^\infty$-superfunctions is not
locally free. Nevertheless, any $G$-supermanifold has an
underlying $G^\infty$-supermanifold.

In the case of $GH^\infty$-supermanifolds (including
$H^\infty$-ones), spaces of values of $GH^\infty$-superfunctions
at different points are not mutually isomorphic because the
Grassmann algebra $\La$ is not a free module with respect to its
subalgebra $\La'$. By these reasons, we turn to
$G$-supermanifolds. Their definition repeats Definition \ref{+27}.

\begin{definition} \label{+30} \mar{+30}
An $(n,m)$-dimensional $G$-supermanifold is a graded local-ringed
space $(M,G_M)$, satisfying the following conditions:

$\bullet$ $M$ is a paracompact topological space;

$\bullet$ $(M,G_M)$ is locally isomorphic to
$(B^{n,m},\ccG_{n,m})$;

$\bullet$ there exists a morphism of sheaves of graded commutative
$\La$-algebras $\dl:G_M\to C^\La_M$, where
\be
C^\La_M\cong C^0_M\ot\La
\ee
is sheaf of continuous $\La$-valued functions on $M$,  and $\dl$
is locally isomorphic to the evaluation morphism (\ref{+22}).
\end{definition}

\begin{example} \label{+50} \mar{+50}
The triple $(B^{n,m},\ccG_{n,m},\dl)$, where $\dl$ is the
evaluation morphism (\ref{+22}),  is called the standard
$G$-supermanifold. For any open subset $U\subset B^{n,m}$, the
space $\ccG_{n,m}(U)$ can be provided with the topology which
makes it into a graded Fr\'echet algebra. Then there are
isometrical isomorphisms
\mar{+65}\ben
&& \ccG_{n,m}(U)\cong \cH^\infty(U)\ot \La\cong
C^\infty(\si^{n,m}(U))\ot\La\ot\w
\Bbb R^m \cong \label{+65} \\
&&\qquad C^\infty(\si^{n,m}(U))\ot\w\Bbb R^{N+m}. \nonumber
\een
\end{example}

\begin{remark} \label{+81} \mar{+81}
Any $GH^\infty$-supermanifold $(M,GH^\infty_M)$ with the structure
sheaf $GH^\infty_M$ is naturally extended to the $G$-supermanifold
$(M,GH^\infty_M\ot \La)$. Every $G$-supermanifold defines an
underlying $G^\infty$-supermanifold $(M,\dl(G_M))$, where
$\dl(G_M)=G^\infty_M$ is the sheaf of $G^\infty$-superfunctions on
$M$.
\end{remark}

As in the case of smooth supermanifolds, the underlying space $M$
of a $G$-supermanifold $(M,G_M)$ is provided with the structure of
a real smooth manifold of dimension $2^{N-1}(n+m)$, and morphisms
of $G$-supermanifolds are smooth morphisms of the underlying
smooth manifolds. However, it may happen that non-isomorphic
$G$-supermanifold have isomorphic underlying smooth manifolds.

\begin{remark} \label{+51} \mar{+51}
Let us present briefly the axiomatic approach to supermanifolds
which enables one to obtain all the previously known types of
supermanifolds in terms of $R^\infty$-supermanifolds
\cite{bart,bart93,bruz99}. This approach to supermanifolds refines
that in \cite{roth}. The $R^\infty$-supermanifolds are introduced
over the above mentioned Arens--Michael algebras of Grassmann
origin \cite{bruz99}, but we here omit the topological side of
their definition, though just the topological properties differ
$R^\infty$-supermanifolds from $R$-supermanifolds in \cite{roth}.

Let $\La$ be a graded commutative algebra of the above mentioned
type (for the sake of simplicity, the reader can think of $\La$ as
being a Grassmann algebra). A superspace over $\La$ is a triple
$(M,\gR^\infty, \dl)$, where $M$ is a paracompact topological
space, $\gR^\infty$ is a sheaf of graded commutative
$\La$-algebras, and $\delta: \gR^\infty\to C^\La_M$ is an
evaluation morphism to the sheaf $C^\La_M$ of continuous
$\La$-valued functions on $M$. Sections of $\gR^\infty$ are called
$R^\infty$-superfunctions. Let $\cM_q$ denote the ideal of the
stalk $\gR_q^\infty$, $q\in M$, formed by the germs of
$R^\infty$-superfunctions $f$ vanishing at a point $q$, i.e., such
that $\dl(f)(q)=0$. An $R^\infty$-supermanifold of dimension
$(n,m)$ is a superspace $(M,\gR^\infty, \dl)$ satisfying the
following four axioms \cite{bruz99}.

\noindent {\bf Axiom 1.} The graded $\gR^\infty$-dual
$(\gd\gR^\infty)^*$ of the sheaf of derivations is a locally free
sheaf of graded $\gR^\infty$-modules of rank $(n,m)$. Every point
$q\in M$ has an open neighborhood $U$ with sections
$x^1,\cdots,x^n\in \gR^\infty(U)_0$, $y^1,\cdots,y^m\in
\gR^\infty(U)_1$ such that $\{dx^i,dy^j\}$ is a graded basis for
$(\gd\gR^\infty)^*(U)$ over $\gR^\infty(U)$.

\noindent {\bf Axiom 2.} Given the above mentioned coordinate
chart, the assignment
\be
q\to (\dl(x^i), \dl(y^j))
\ee
determines a homeomorphism of $U$ onto an open subset of
$B^{n,m}$.

\noindent {\bf Axiom 3.} For every $q\in M$, the ideal $\cM_q$ is
finitely generated.

\noindent {\bf Axiom 4.} For every open subset $U\subset M$, the
topological algebra $\gR^\infty(U)$ is Hausdorff and complete.

An $R$-supermanifold over a graded commutative Banach algebra
satisfying Axiom 4 is an $R^\infty$-supermanifold. The standard
$G$-supermanifold in Example \ref{+50} is an
$R^\infty$-supermanifold. Moreover, in the case of a finite
Grassmann algebra $\La$, the category of $R^\infty$-supermanifolds
and the category of $G$-supermanifolds are equivalent.
\end{remark}

Let $(M,G_M)$ be a $G$-supermanifold. As was mentioned above, it
satisfies the Axioms 1-4. Sections $u$ of the sheaf $\gd G_M$ of
graded derivations are called supervector fields on the
$G$-supermanifold $(M,G_M)$, while sections $\f$ of the dual sheaf
$\gd G_M^*$ are one-superforms on $(M,G_M)$. Given a coordinate
chart $(q^i)=(x^i,y^j)$ on $U\subset M$, supervector fields and
one-superforms read
\be
u=u^i\dr_i, \qquad \f=\f_idq^i,
\ee
where coefficients $u^i$ and $\f_i$ are $G$-superfunctions on
$U$. The graded differential calculus in supervector fields and
superforms obeys the standard formulae (\ref{ws14}), (\ref{ws45}),
(\ref{ws44}) and (\ref{ws46}).

Let us consider cohomology of $G$-supermanifolds. Given a
$G$-supermanifold $(M,G_M)$, let
\be
\gO^k_{\La M}=\gO^k_M\ot\La
\ee
be the sheaves of smooth $\La$-valued exterior forms on $M$. These
sheaves are fine, and they constitute the fine resolution
\be
0\to\La\to C^\infty_M\ot\La \to \gO^1_M\ot\La\to\cdots
\ee
of the constant sheaf $\La$ on $M$. We have the corresponding de
Rham complex
\be
0\to\La\to C^\infty_\La(M) \to \cO^1_\La(M)\to\cdots
\ee
of $\La$-valued exterior forms on $M$. By virtue of the abstract
de Rham theorem \cite{book09,sard09}, the cohomology $H^*_\La(M)$
of this complex is isomorphic to the sheaf cohomology $H^*(M;\La)$
of $M$ with coefficients in the constant sheaf $\La$ and,
consequently, is related to the de Rham cohomology as follows:
\mar{+146}\beq
H^*_\La(M)=H^*(M;\La)=H^*(M)\ot\La. \label{+146}
\eeq
Thus, the cohomology groups of $\La$-valued exterior forms do not
provide us with information on the $G$-supermanifold structure of
$M$.

Let us turn to cohomology of superforms on a $G$-supermanifold
$(M,G_M)$. The sheaves $\op\w^k\gd G_M^*$ of superforms constitute
the complex
\mar{+141}\beq
0\to\La\to G_M\to \gd^*G_M\to \cdots. \label{+141}
\eeq
The Poincar\'e lemma for superforms is proved to hold
\cite{bart,bruz88}, and this complex is exact. However, the
structure sheaf $G_M$ need not be acyclic, and the exact sequence
(\ref{+141}) fails to be a resolution of the constant sheaf $\La$
on $M$ in general. Therefore, the cohomology $H^*_S(M)$ of the de
Rham complex of superforms are not equal to cohomology
$H^*(M;\La)$ of $M$ with coefficients in the constant sheaf $\La$,
and need not be related to the de Rham cohomology $H^*(M)$ of the
smooth manifold $M$. In particular, cohomology $H^*_S(M)$ is not a
topological invariant, but it is invariant under $G$-isomorphisms
of $G$-supermanifolds.

\begin{prop} \label{+140} \mar{+140}  The structure sheaf
$\ccG_{n,m}$ of the standard $G$-supermanifold
$(B^{n,m},\ccG_{n,m})$ is acyclic, i.e.,
\be
H^{k>0}(B^{n,m};\ccG_{n,m})=0.
\ee
\end{prop}

The proof is based on the isomorphism (\ref{+65}) and some
cohomological constructions \cite{bart,bruz99}.

\section{DeWitt supermanifolds}

There exists a particular class of supermanifolds, called DeWitt
supermanifolds. Their notion implies that a supervector space
$B^{n,m}$ is provided with the DeWitt topology, which is coarser
than the Euclidean one. This is the coarsest topology such that
the body map
\mar{w1002}\beq
\si^{n,m}:B^{n,m}\to\Bbb R^n \label{w1002}
\eeq
is continuous. The open sets in the DeWitt topology are of the
form $V\times \cR^{n,m}$, where $V$ are open sets in $\Bbb R^n$.
Clearly, this topology is not Hausdorff.

\begin{definition} \label{+31} \mar{+31}
A smooth supermanifold (resp. a $G$-supermanifold) is said to be a
DeWitt supermanifold if it admits an atlas such that the local
morphisms $\f_\zeta: U_\zeta\to B^{n,m}$ in Definition \ref{+27}
(resp. Definition \ref{+30}) are continuous with respect to the
DeWitt topology, i.e., $\f_\zeta(U_\zeta)\subset B^{n,m}$ are open
in this topology.
\end{definition}

Given an atlas $(U_\zeta,\f_\zeta)$ of a DeWitt supermanifold in
accordance with Definition \ref{+31}, it is readily observed that
its transition functions $\f_\zeta\circ\f_\xi^{-1}$ must preserve
the fibration $\si^{n,m}$ (\ref{w1002}) whose fibre
$(\si^{n,m})^{-1}(z)$ over $z\in \Bbb R^n$ is equipped with the
coarsest topology, where only $\emptyset$ and
$(\si^{n,m})^{-1}(z)$ are open sets. This fact leads to the
following.

Every DeWitt supermanifold is a locally trivial topological fibre
bundle $\si_M: M\to Z_M$ over an $n$-dimensional smooth manifold
$Z_M$ with the typical fibre $R^{n,m}=\si(B^{n,m})$. The base
$Z_M$ of this fibre bundle is said to be a body manifold of a
DeWitt supermanifold, while the surjection $\si_M$ is called a
body map.

There is the above mentioned correspondence between the graded
manifolds and the DeWitt $H^\infty$-supermanifolds. It is based on
the following facts.

(i) The structure sheaf $\gA$ of a graded manifold $(Z,\gA)$ is
locally isomorphic to the sheaf $C^\infty_U\ot\w\Bbb R^m$.

(ii) Given a DeWitt $H^\infty$-supermanifold $(M,H^\infty_M)$, the
direct image $\si_*(H^\infty_M)$ of its structure sheaf onto the
body manifold $Z_M$ is locally isomorphic to the sheaf
$C^\infty_U\ot \w\Bbb R^m$. The expression (\ref{+41}) provides
this isomorphism in an explicit form.

(iii) The graded local-ringed spaces $(M,H^\infty_M)$ and
$(Z_M,\si_*(H^\infty_M))$ determine the same element of the
cohomology set $H^1(Z_M;\Is(\w\Bbb R^m)^\infty)$.

Thus, we come to the following statement \cite{bart,batch2}.

\begin{theorem}\label{+44} \mar{+44}
Given a DeWitt $H^\infty$-supermanifold $(M,H^\infty_M)$, the
associated pair $(Z_M,\si_*(H^\infty_M))$ is a graded manifold.
Conversely, for any graded manifold $(Z,\gA)$, there exists a
DeWitt $H^\infty$-supermanifold, whose body manifold is $Z$ and
whose structure sheaf $\gA$ is isomorphic to $\si_*(H^\infty_M)$.
\end{theorem}

Then by virtue of Batchelor's Theorem \ref{lmp1} and Theorem
\ref{+44}, there is one-to-one correspondence between the
isomorphism classes of DeWitt $H^\infty$-supermanifolds of odd
rank $m$ with a body manifold $Z$ and the equivalence classes of
$m$-dimensional vector bundles over $Z$. This result is extended
to DeWitt $GH^\infty$-, $G^\infty$- and $G$-supermanifolds because
their isomorphism classes
   are in one-to-one correspondence with isomorphism classes of DeWitt
$H^\infty$-supermanifolds \cite{bart}.

Let us say something more on DeWitt $G$-supermanifolds.

\begin{prop} \label{+145} \mar{+145} The structure sheaf
$G_M$ of a DeWitt $G$-supermanifold is acyclic, and so is any
locally free sheaf of graded $G_M$-modules  \cite{bart,bruz99}.
\end{prop}

\begin{prop} \label{1035} \mar{1035} There is an isomorphism of
the cohomology $H^*_S(M)$ of superforms on a DeWitt
$G$-supermanifold to the cohomology (\ref{+146}) of $\La$-valued
exterior forms on its body manifold $Z_M$, i.e.,
$H^*_S(M)=H^*(Z_M)\ot\La$ \cite{rab}.
\end{prop}

These results are based on the fact that the structure sheaf $G_M$
on $M$, provided with the DeWitt topology, is fine. However, this
does not imply automatically that $G_M$ is acyclic since the
DeWitt topology is not paracompact. Nevertheless, it follows that
the image $\si_*(G_M)$ of $G_M$ on the body manifold $Z_M$ is fine
and acyclic. Then Proposition \ref{+140} lead to Proposition
\ref{+145}. In particular, the sheaves of superforms on a DeWitt
$G$-supermanifold are acyclic. Then they constitute the resolution
of the constant sheaf $\La$ on $M$, and we obtain isomorphisms
\be
H^*_S(M)= H^*(M;\La)=H^*(M)\ot\La.
\ee
Since the typical fibre of the fibre bundle $M\to Z_M$ is
contractible, then $H^*(M)=H^*(Z_M)$ such that the isomorphism in
Proposition \ref{1035} takes place.

\section{Supervector bundles}

Supervector bundles are considered in the category of
$G$-supermanifolds. We start with the definition of the product of
two $G$-supermanifolds seen as a trivial supervector bundle.

Let $(B^{n,m},\ccG_{n,m})$ and $(B^{r,s},\ccG_{r,s})$ be two
standard $G$-supermanifolds in Example \ref{+30}. Given open sets
$U\subset B^{n,m}$ and $V\subset B^{r,s}$, we consider the
presheaf
\mar{+67}\beq
U\times V\to \ccG_{n,m}(U)\wh\ot \ccG_{r,s}(V), \label{+67}
\eeq
where $\wh\ot$ denotes the tensor product of modules completed in
Grothendieck's topology. Due to the isomorphism (\ref{+65}), it is
readily observed that the structure sheaf $\ccG_{n+r,m+s}$ of the
standard $G$-supermanifold on $B^{n+r,m+s}$ is isomorphic to that,
defined by the presheaf (\ref{+67}). This construction is
generalized to arbitrary $G$-supermanifolds as follows.

Let  $(M,G_M)$ and $(M',G_{M'})$ be two $G$-supermanifolds of
dimensions $(n,m)$ and $(r,s)$, respectively. Their product
\be
(M,G_M) \times(M',G_{M'})
\ee
is defined as the graded local-ringed space $(M\times M',
G_M\wh\ot G_{M'})$, where $G_M\wh\ot G_{M'}$ is the sheaf
determined by the presheaf
\be
&& U\times U'\to G_M(U)\wh\ot G_{M'}(U'),\\
&& \dl: G_M(U)\wh\ot G_{M'}(U')\to C^\infty_{\si(U)}\wh\ot
C^\infty_{\si(U')}= C^\infty_{\si_M(U)\times \si_M(U')},
\ee
for any open subsets $U\subset M$ and $U'\subset M'$. This product
is a $G$-supermanifold of dimension $(n+r,m+s)$ \cite{bart}.
Furthermore, there is the epimorphism
\be
  \pr_1:(M,G_M) \times(M',G_{M'})\to (M,G_M).
\ee
One may define its section over an open subset $U\subset M$ as the
$G$-supermanifold morphism
\be
s_U:(U,G_M|_U)\to (M,G_M) \times(M',G_{M'})
\ee
such that $\pr_1\circ s_U$ is the identity morphism of
$(U,G_M|_U)$. Sections $s_U$ over all open subsets $U\subset M$
determine a sheaf on $M$. This sheaf should be provided with a
suitable graded commutative $G_M$-structure.

For this purpose, let us consider the product
\mar{+72}\beq
(M,G_M)\times (B^{r\mid s}, \ccG_{r\mid s}), \label{+72}
\eeq
  where $B^{r\mid s}$
is the superspace (\ref{+70}). It is called a  product
$G$-supermanifold. Since the $\La_0$-modules $B^{r\mid s}$ and
$B^{r+s,r+s}$ are isomorphic, $B^{r\mid s}$ has a natural
structure of an $(r+s,r+s)$-dimensional $G$-supermanifold. Because
$B^{r\mid s}$ is a free graded $\La$-module of the type $(r,s)$,
the sheaf $S_M^{r\mid s}$ of sections of the fibration
\mar{+71}\beq
(M,G_M)\times (B^{r\mid s}, \ccG_{r\mid s})\to (M,G_M) \label{+71}
\eeq
has the structure of the sheaf of  free graded $G_M$-modules of
rank $(r,s)$. Conversely, given a $G$-supermanifold $(M,G_M)$ and
a sheaf $S$ of
  free graded $G_M$-modules of rank
$(r,s)$ on $M$, there exists a product $G$-supermanifold
(\ref{+72}) such that $S$ is isomorphic to the sheaf of sections
of the fibration (\ref{+71}).

Let us turn now to the notion of a supervector bundle over
$G$-supermanifolds. Similarly to smooth vector bundles
\cite{book09,sard09}, one can require of the category of
supervector bundles over $G$-supermanifolds to be equivalent to
the category of locally free sheaves of graded modules on
$G$-supermanifolds. Therefore, we can restrict ourselves to
locally trivial supervector bundles with the standard fibre
$B^{r\mid s}$.

\begin{definition} \label{+78} \mar{+78}
A supervector bundle over a $G$-supermanifold $(M,G_M)$ with the
standard fibre $(B^{r\mid s},\ccG_{r\mid s})$ is defined as a pair
$((Y,G_Y),\pi)$ of a $G$-supermanifold $(Y,G_Y)$ and a
$G$-epimorphism
\mar{+76}\beq
\pi: (Y,G_Y)\to (M,G_M) \label{+76}
\eeq
such that $M$ admits an atlas $\{(U_\zeta,\psi_\zeta\}$ of local
$G$-isomorphisms
\be
\psi_\zeta: (\pi^{-1}(U_\zeta), G_Y|_{\pi^{-1}(U_\zeta)})\to
(U_\zeta, G_M|_{U_\zeta})\times (B^{r\mid s},\ccG_{r\mid s}).
\ee
\end{definition}

It is clear that sections of the supervector bundle (\ref{+76})
constitute a  locally free sheaf of graded $G_M$-modules. The
converse of this fact is the following \cite{bart}.

\begin{theorem} \label{+77} \mar{+77}
For any locally free sheaf $S$ of graded $G_M$-modules of rank
$(r,s)$ on a $G$-supermanifold $(M,G_M)$, there exists a
supervector bundle over $(M,G_M)$ such that $S$ is isomorphic to
the structure sheaf of its sections.
\end{theorem}

The fibre $Y_q$, $q\in M$, of the supervector bundle in Theorem
\ref{+77} is the quotient
\be
S_q/\cM_q\cong S_{Mq}^{r\mid s}/(\cM_q\cdot S_{Mq}^{r\mid s})\cong
B^{r\mid s}
\ee
of the stalk $S_q$ by the submodule $\cM_q$ of the germs $s\in
S_q$ whose evaluation $\dl(f)(q)$ vanishes. This fibre is a graded
$\La$-module isomorphic to $B^{r\mid s}$, and is provided with the
structure of the standard $G$-supermanifold.

\begin{remark} \label{+79} \mar{+79}
The proof of Theorem \ref{+77} is based on the fact that, given
the transition functions $\rho_{\zeta\xi}$ of the sheaf $S$, their
evaluations
\mar{+80}\beq
g_{\zeta\xi}=\dl(\rho_{\zeta\xi}) \label{+80}
\eeq
  define the morphisms
\be
U_\zeta\cap U_\xi \to  GL(r|s;\La),
\ee
and they are assembled into a cocycle of $G^\infty$-morphisms from
$M$ to the general linear graded  group $GL(r|s;\La)$. Thus, we
come to the notion of a $G^\infty$-vector bundle. Its definition
is a repetition of Definition \ref{+78} if one replaces
$G$-supermanifolds and $G$-morphisms with the $G^\infty$- ones.
Moreover, the $G^\infty$-supermanifold underlying a supervector
bundle (see Remark \ref{+81}) is a $G^\infty$-supervector bundle,
whose transition functions $g_{\zeta\xi}$ are related to those of
the supervector bundle by the evaluation morphisms (\ref{+80}),
and are $GL(r|s;\La)$-valued transition functions.
\end{remark}

Since the category of supervector bundles over a $G$-supermanifold
$(M,G_M)$ is equivalent to the category of locally free sheaves of
graded $G_M$-modules, one can define the usual operations of
direct sum, tensor product, etc. of supervector bundles.

Let us note that any supervector bundle admits the canonical
global zero section. Any section of the supervector bundle $\pi$
(\ref{+76}), restricted to its trivialization chart
\mar{+156}\beq
(U, G_M\mid_U)\times (B^{r\mid s},\ccG_{r\mid s}), \label{+156}
\eeq
is represented by a sum $s = s^a(q)\e_a$, where $\{\e_a\}$ is the
basis  for the graded $\La$-module $B^{r\mid s}$, while $s^a(q)$
are $G$-superfunctions on $U$. Given another trivialization chart
$U'$ of $\pi$, a transition function
\mar{+155}\beq
s'^b(q)\e'_b=s^a(q)h^b{}_a(q)\e_b, \qquad q\in U\cap U',
\label{+155}
\eeq
is given by the $(r+s)\times(r+s)$ matrix $h$ whose entries
$h^b{}_a(q)$ are $G$-superfunctions on $U\cap U'$. One can think
of this matrix as being a section of the supervector bundle over
$U\cap U$ with the above mentioned group $GL(r|s;\La)$ as a
typical fibre.

\begin{example} \label{+84} \mar{+84}
Given a $G$-supermanifold $(M,G_M)$, let us consider the locally
free sheaf $\gd G_M$ of graded derivations of $G_M$. In accordance
with Theorem \ref{+77}, there is a supervector bundle $T(M,G_M)$,
called supertangent bundle, whose structure sheaf is isomorphic to
$\gd G_M$. If $(q^1,\ldots,q^{m+n})$ and $(q'^1,\ldots,q'^{m+n})$
are two coordinate charts on $M$, the Jacobian matrix
\be
h^i_j=\frac{\dr q'^i}{\dr q^j}, \qquad i,j=1,\ldots, n+m,
\ee
(see the prescription (\ref{+83})) provides the transition
morphisms for $T(M,G_M)$.

It should be emphasized that the underlying $G^\infty$-vector
bundle of the supertangent bundle $T(M,G_M)$, called
$G^\infty$-supertangent bundle, has the transition functions
$\dl(h^i_j)$ which cannot be written as the Jacobian matrices
since the derivatives of $G^\infty$-superfunctions with respect to
odd arguments are ill defined and the sheaf $\gd G^\infty_M$ is
not locally free.
\end{example}

\section{Superconnections}

Given a supervector bundle $\pi$ (\ref{+76}) with the structure
sheaf $S$, one can introduce a connection on this supervector
bundle  as a splitting of the the exact sequence of sheaves
\mar{+121}\beq
0\to \gd G_M^*\ot S\to (G_M\oplus\gd G_M^*)\ot S\to S\to 0
\label{+121}
\eeq
\cite{book00}. Its splitting is an even sheaf morphism
\mar{+122}\beq
\nabla: S\to \gd^*G_M\ot S \label{+122}
\eeq
satisfying the Leibniz rule
\mar{+158}\beq
\nabla (fs)=df\ot s + f\nabla(s), \qquad f\in G_M(U), \qquad s\in
S(U), \label{+158}
\eeq
for any open subset $U\in M$. The sheaf morphism (\ref{+122}) is
called a superconnection on the supervector bundle $\pi$
(\ref{+76}). Its curvature is given by the expression
\mar{+157}\beq
R=\nabla^2:S\to \op\w^2\gd G^*_M\ot S. \label{+157}
\eeq

The exact sequence (\ref{+121}) need not be split. One can apply
the criterion in Section 1.8 in order to study the existence of a
superconnection on supervector bundles. Namely, the exact sequence
(\ref{+121}) leads to the exact sequence of sheaves
\be
0\to \hm(S,\gd G_M^*\ot S)\to \hm(S,(G_M\oplus\gd G_M^*)\ot S) \to
\hm(S,S)\to 0
\ee
and to the corresponding exact sequence of the cohomology groups
\be
&& 0\to H^0(M; \hm(S,\gd G_M^*\ot S)) \to H^0(M;
\hm(S,(G_M\oplus\gd G_M^*)\ot S)) \\
&& \qquad \to H^0(M;\hm(S,S))\to H^1(M;\hm(S,\gd G_M^*\ot S))\to
\cdots.
\ee
The exact sequence  (\ref{+121}) defines the Atiyah class
\be
{\rm At}(\pi)\in H^1(M;\hm(S,\gd G_M^*\ot S))
\ee
 of the supervector bundle $\pi$ (\ref{+76}). If
the Atiyah class vanishes, a superconnection on this supervector
bundle exists. In particular, a superconnection exists if the
  cohomology set $H^1(M;\hm(S,\gd G_M^*\ot S))$ is trivial.
In contrast with the sheaf of smooth functions, the structure
sheaf $G_M$ of a $G$-supermanifold is not acyclic in general,
cohomology $H^*(M;\hm(S,\gd G_M^*\ot S))$ is not trivial, and a
supervector bundle need not admit a superconnection.

\begin{example} \label{+150} \mar{+150}
In accordance with Proposition \ref{+140}, the structure sheaf of
the standard $G$-supermanifold $(B^{n,m},\ccG_{n,m})$ is acyclic,
and the trivial supervector bundle
\mar{+151}\beq
(B^{n,m},\ccG_{n,m})\times (B^{r\mid s},\ccG_{r\mid s}) \to
(B^{n,m},\ccG_{n,m}) \label{+151}
\eeq
has obviously a superconnection, e.g., the trivial
superconnection.
\end{example}

\begin{example} By virtue of Proposition \ref{+145}, the structure sheaf
of a DeWitt $G$-supermanifold $(M,G_M)$ is acyclic, and so is the
sheaf $\hm(S,\gd G_M^*\ot S)$. It follows that any supervector
bundle over a DeWitt $G$-supermanifold admits a superconnection.
\end{example}

Example \ref{+150} enables one to obtain a local coordinate
expression for a superconnection on a supervector bundle $\pi$
(\ref{+76}), whose typical fibre is $B^{r\mid s}$ and whose base
is a $G$-supermanifold locally isomorphic to the standard
$G$-supermanifold $(B^{n,m},\ccG_{n,m})$. Let $U\subset M$
(\ref{+156}) be a trivialization chart of this supervector bundle
such that every section $s$ of $\pi|_U$ is represented by a sum
$s^a(q)\e_a$, while the
  sheaf of one-superforms
$\gd^* G_M|_U$ has a local basis $\{d q^i\}$.
  Then a
superconnection $\nabla$ (\ref{+122}) restricted to this
trivialization chart can be given by a collection of coefficients
$\nabla_i{}^a{}_b$:
\mar{+160}\beq
\nabla (\e_a)=dq^i\ot (\nabla_i{}^b{}_a\e_b), \label{+160}
\eeq
which  are $G$-superfunctions on $U$. Bearing in mind the Leibniz
rule (\ref{+158}), one can compute the coefficients of the
curvature form (\ref{+157}) of the superconnection (\ref{+160}).
We have
\be
&& R(\e_a)=\frac12 dq^i\w dq^j\ot R_{ij}{}^b{}_a\e_b, \\
&& R_{ij}{}^a{}_b =(-1)^{[i][j]}\dr_i\nabla_j{}^a{}_b -\dr_j\nabla_i{}^a{}_b
+ (-1)^{[i]([j]+[a]+[k])}\nabla_j{}^a{}_k\nabla_i{}^k{}_b -\\
&& \qquad (-1)^{[j]([a]+[k])}\nabla_i{}^a{}_k\nabla_j{}^k{}_b.
\ee
In a similar way, one can obtain the transformation law of the
superconnection coefficients (\ref{+160}) under the transition
morphisms (\ref{+155}). In particular, any trivial supervector
bundle admits the trivial superconnection $\nabla_i{}^b{}_a=0$.

\section{Principal superconnections}

In contrast with a supervector bundle, the  structure sheaf $G_P$
of a principal superbundle $(P,G_P)\to (M,G_M)$ is not a sheaf of
locally free $G_M$-modules in general. Therefore, the above
technique of connections on modules and sheaves is not applied to
principal superconnections in a straightforward way. Principal
superconnections are  introduced on principal superbundles by
analogy with principal connections on smooth principal bundles
\cite{bart}. For the sake of simplicity, let us denote
$G$-supermanifolds $(M,G_M)$ and their morphisms
\be
(\vf:M\to N, \qquad \Phi:G_N\to \vf_*(M))
\ee
by $\wh M$ and $\wh \vf$, respectively. Given a point $q\in M$, by
$\wh q=(q,\La)$ is meant the trivial $G$-supermanifold of
dimension $(0,0)$. We will start with the notion of a $G$-Lie
supergroup $\wh H$. The relations between $G$- $GH^\infty$- and
$G^\infty$-Lie supergroups follow the relations between the
corresponding classes of superfunctions.

\begin{definition} \label{+205} \mar{+205}
A $G$-supermanifold $\wh H=(H,\cH)$ is said to be a $G$-Lie
supergroup if there exist the following $G$-supermanifold
morphisms:

$\bullet$ a multiplication $\wh m:\wh H\times \wh H\to\wh H$,

$\bullet$ a unit $\wh\ve: \wh e\to \wh H$,

$\bullet$ an inverse $\wh  S:\wh H\to\wh H$,

\noindent together with the natural identifications
\be
\wh e\times \wh H=\wh H\times \wh e= \wh H,
\ee
which satisfy the associativity
\be
\wh m\circ(\Id \times \wh m)=\wh m\circ(\wh m\times\Id):\wh
H\times\wh H\times\wh H\to \wh H\times\wh H\to \wh H,
\ee
the unit property
\be
(\wh m\circ (\wh\ve\times\Id))(\wh e\times \wh H)= (\wh m\circ
(\Id\times\wh\ve))(\wh H\times \wh e)=\id H,
\ee
and the inverse property
\be
(\wh m\circ (\wh S,\Id))(\wh H)=(\wh m\circ (\Id,\wh S))(\wh
H)=\wh\ve(\wh e).
\ee
\end{definition}

Given a point $g\in H$, let us denote by $\wh g:\wh e\to\wh H$ the
$G$-supermanifold morphism whose range in $H$ is $g$. Then one can
introduce the notions of the left translation $\wh L_g$ and the
right translation $\wh R_g$ as the $G$-supermanifold morphisms
\be
&&\wh L_g: \wh H=\wh e\times \wh H\ar^{\wh g\times\Id} \wh H\times\wh H\ar^{\wh
m} \wh H,\\
&&\wh R_g: \wh H=\wh H\times \wh e\ar^{\Id\times\wh g} \wh H\times\wh H\ar^{\wh
m} \wh H.
\ee

\begin{remark} \label{+202} \mar{+202}
Given a $G$-Lie supergroup $\wh H$, the underlying smooth manifold
$H$ is provided with the structure of a real Lie group of
dimension $2^{N-1}(n+m)$, called the underlying Lie group. In
particular, the actions on the underlying Lie group $H$,
corresponding to the left and right translations by $\wh g$, are
ordinary left and right translations by $g$.
\end{remark}

Let us reformulate the group axioms in Definition \ref{+205} in
terms of the structure sheaf $\cH$ of the $G$-Lie supergroup
$(H,\cH)$. We observe that $\cH$ has properties of a sheaf of
graded Hopf algebras as follows.

If $(H,\cH)$ is a $G$-Lie supergroup, the structure sheaf $\cH$ is
provided with the sheaf morphisms:

$\bullet$ a comultiplication $\wh m^*:\cH\to m_*(\cH\wh\ot\cH)$,

$\bullet$ a counit $\wh\ve^*:\cH\to e_*(\La)$,

$\bullet$ a coinverse $\wh S: \cH\to s_*\cH$.

\noindent Let us denote
\be
k=m\circ(\Id\times m)=m\circ (m\times\Id): H\times H\times H\to H.
\ee
Then the group axioms in Definition \ref{+205} are equivalent to
the relations
\be
&& ((\Id\ot\wh m^*)\circ\wh m^*)(\cH)=((\wh m^*\ot\Id)\circ\wh m^*)(\cH)=
k_*(\cH\wh\ot\cH\wh\ot\cH),\\
&& (\wh m^*\circ (\Id\ot\wh\ve^*))(\cH\wh\ot e_*(\La))=
(\wh m^*\circ (\wh\ve^*\ot\Id))(e_*(\La)\wh\ot\cH)=\id\cH,\\
&& (\id\cdot\wh S^*)\circ\wh m^*=(\wh S^*\cdot\Id)\circ\wh m^*=\wh\ve^*.
\ee
Comparing these relations with the axioms of a Hopf algebra in
Section 10.2, one can think of the structure sheaf of a $G$-Lie
group as being a sheaf of graded topological Hopf algebras.

\begin{example}
The general linear graded group $GL(n|m;\La)$ is endowed with the
natural structure of an $H^\infty$-supermanifold of dimension
$(n^2+m^2, 2nm)$. The matrix multiplication gives the
$H^\infty$-morphism
\be
m: GL(n|m;\La)\times GL(n|m;\La)\to GL(n|m;\La)
\ee
such that $F(g,g')\mapsto F(gg')$. It follows that $GL(n|m;\La)$
is an $H^\infty$-Lie supergroup. It is trivially extended to the
$G$-Lie supergroup $\wh{GL}(n\mid m;\La)$, called the general
linear supergroup.
\end{example}

A Lie superalgebra ${\got h}$ of a $G$-Lie supergroup $\wh H$ is
defined as an algebra of left-invariant supervector fields on $\wh
H$. Let us recall that a supervector field $u$ on a
$G$-supermanifold $\wh H$ is a derivation of its structure sheaf
$\cH$. It is called left-invariant if
\be
(\Id\ot u)\circ \wh m^* =\wh m^*\circ u.
\ee
If $u$ and $u'$ are left-invariant supervector fields, so are
$[u,u']$ and $au+a'u'$, $a,a'\in\La$. Hence, left-invariant
supervector fields constitute a Lie superalgebra. The Lie
superalgebra ${\got h}$ can be identified with the supertangent
space $T_e(\wh H)$. Moreover, there is the sheaf isomorphism
\mar{+229}\beq
\cH\ot {\got h}= \gd\cH, \label{+229}
\eeq
i.e., the sheaf of supervector fields on a $G$-Lie supergroup $\wh
H$ is the globally free  sheaf of graded $\cH$-modules of rank
$(n,m)$, generated by left-invariant supervector fields. The Lie
superalgebra of right-invariant supervector fields on $\wh H$ is
introduced in a similar way.

Let us consider the right action of a $G$-Lie supergroup $\wh H$
on a $G$-supermanifold $\wh P$. This is a $G$-morphism
\be
\wh\rho:\wh P\times \wh H\to \wh P
\ee
such that
\be
&& \wh\rho\circ (\wh\rho\times\Id)=\wh\rho\circ(\Id\times\wh m):\wh P\times
\wh H\times \wh H\to \wh P,\\
&& \wh\rho\circ (\Id\times \wh \ve)(\wh P\times \wh e)=\id\wh P.
\ee
The left action of $\wh H$ on $\wh P$ is defined similarly.

\begin{example}
Obviously, a $G$-Lie supergroup acts on itself both on the left
and on the right by the multiplication morphism $\wh m$.

The general linear supergroup $\wh{GL}(n|m;\La)$ acts linearly on
the standard supermanifold $B^{n\mid n}$ on the left by the matrix
multiplication which is a $G$-morphism.
\end{example}

Let $\wh P$ and $\wh P'$ be $G$-supermanifolds that are acted on
by the same $G$-Lie supergroup $\wh H$. A $G$-supermanifold
morphism $\wh\vf:\wh P\to \wh P'$ is said to be $\wh H$-invariant
if
\be
\wh\vf\circ\wh\rho= \wh\rho'\circ(\wh\vf\times\Id): \wh P\times\wh
H\to \wh P'.
\ee

\begin{definition} \label{+224} \mar{+224}
A quotient of an action of a $G$-Lie supergroup on a
$G$-submanifold $\wh P$ is a pair $(\wh M, \wh \pi)$ of a
$G$-supermanifold $\wh M$ and a $G$-supermanifold morphism $\wh
\pi:\wh P\to\wh M$ such that:

(i) there is the equality
\mar{+222}\beq
\wh\pi\circ\wh\rho=\wh\pi\circ\wh\pr_1:\wh P\times\wh H\to\wh
M,\label{+222}
\eeq

(ii) for every morphism $\wh\vf:\wh P\to \wh M'$ such that
$\wh\vf\circ\wh\rho=\wh\vf\circ\wh\pr_1$, there is a unique
morphism $\wh g:\wh M\to\wh M'$ with $\wh\vf=\wh g\circ \wh\pi$.
\end{definition}

The quotient $(\wh M,\wh\pi)$ does not   necessarily exists. If it
exists, there is a monomorphism  of the structure sheaf $G_M$ of
$\wh M$ into the direct image $\pi_*G_P$. Since the $G$-Lie group
$\wh H$ acts trivially on $\wh M$, the range of this monomorphism
is a subsheaf of $\pi_*G_P$, invariant under the action of $\wh
H$. Moreover, there is an isomorphism
\mar{+221}\beq
G_M\cong (\pi_*G_P)^H \label{+221}
\eeq
between $G_M$ and the subsheaf of $G_P$ of $\wh H$-invariant
sections. The latter is generated by sections of $G_P$ on
$\pi^{-1}(U)$, $U\subset M$, which are $\wh H$-invariant as
$G$-morphisms $\wh U\to \La$, where one takes the trivial action
of $\wh H$ on $\La$.

Let us denote the morphism in the equality (\ref{+222}) by $\vt$.
It is readily observed that the invariant sections of
$G_P(\pi^{-1}(U))$ are exactly the elements which have the same
image under the morphisms
\be
&&\wh\rho^*:G_P(\pi^{-1}(U))\to (\cH\wh\ot G_P)(\vt^{-1}(U)),\\
&& \wh\pr_1^*:G_P(\pi^{-1}(U))\to (\cH\wh\ot G_P)(\vt^{-1}(U)).
\ee
Then the isomorphism (\ref{+221}) leads to the exact sequence of
sheaves of $\La$-modules on $M$
\mar{+223}\beq
0\ar G_M\ar^{\wh\pi^*} \pi_*G_P\ar^{\wh\rho^*-\wh\pr_1^*}
\vt_*(G_M\wh\ot \cH). \label{+223}
\eeq

\begin{definition} \label{+225} \mar{+225}
A principal superbundle of a $G$-Lie supergroup $\wh H$ is defined
as a locally trivial quotient $\pi: \wh P\to\wh M$, i.e., there
exists an open covering $\{U_\zeta\}$ of $M$ together with $\wh
H$-invariant isomorphisms
\be
\wh\psi_\zeta: \wh P\mid_{\wh U_\zeta}\to \wh U_\zeta\times \wh H,
\ee
where $\wh H$ acts on
\mar{+226}\beq
\wh U_\zeta\times \wh H\to \wh U_\zeta \label{+226}
\eeq
  by the right
multiplication.
\end{definition}

\begin{remark}
In fact, we need only the condition (i) in Definition \ref{+224}
of the action of $\wh H$ on $\wh P$ and the condition of local
triviality of $\wh P$.
\end{remark}

In an equivalent way, one can think of a principal superbundle as
being glued together of trivial principal superbundles
(\ref{+226}) by $\wh H$-invariant transition functions
\be
\wh \f_{\zeta\xi}: \wh U_{\zeta\xi}\times \wh H\to \wh
U_{\zeta\xi}\times \wh H, \qquad U_{\zeta\xi}=U_\zeta\cap U_\xi,
\ee
which fulfill the cocycle condition.

As in the case of smooth principal bundles, the following two
types of supervector fields on a principal superbundle are
introduced.

\begin{definition}
A supervector field $u$ on a principal superbundle $\wh P$ is said
to be invariant if
\be
\wh\rho^*\circ u=(u\ot \Id)\circ u: G_P\to \rho_*(G_P\wh\ot\cH).
\ee
\end{definition}

One can associate to every open subset $V\subset M$ the
$G_M(V)$-module of all $\wh H$-invariant supervector fields on
$\pi^{-1}(V)$, thus defining the sheaf $\gd^H(\pi_*G_P)$ of
$G_M$-modules.

\begin{definition}
A fundamental supervector field $\wt \up$ associated to an element
$\up\in {\got h}$ is defined by the condition
\be
\wt\up=(\Id\ot \up)\circ \wh\rho^*: G_P\to G_P\wh\ot e_*(\La)=G_P.
\ee
\end{definition}

Fundamental supervector fields generate the sheaf $\cV G_P$ of
$G_P$-modules of vertical supervector field on the principal
superbundle $\wh P$, i.e., $u\circ\pi^*=0$. Moreover, there is an
isomorphism of sheaves of $G_P$-modules
\be
G_P\ot{\got h}\ni F\ot\up\mapsto F\wt\up\in \cV G_P,
\ee
which is similar to the isomorphism (\ref{+229}).

Let us consider the sheaf
\be
(\pi_*\cV G_P)^H=\pi_*(\cV G_P)\cap \gd^H(\pi_*G_P)
\ee
on $M$ whose sections are vertical $\wh H$-invariant supervector
fields.

\begin{prop} \cite{bart}. There is the exact sequence of sheaves of
$G_M$-modules
\mar{+228}\beq
0\to (\pi_*\cV G_P)^H\to \gd^H(\pi_*G_P)\to \gd G_M\to 0.
\label{+228}
\eeq
\end{prop}

The exact sequence (\ref{+228}) is similar to the exact sequence
of sheaves of $C^\infty$-modules
\be
0\to (V_GP)_X\to (T_GP)_X\to \gd C^\infty_X\to 0
\ee
in the case of smooth principal bundles. Accordingly, we come to
the following definition of a superconnection on a principal
superbundle.

\begin{definition}
A superconnection on a principal superbundle $\pi: \wh P\to\wh M$
(or simply a principal superconnection) is defined as a splitting
\mar{+230}\beq
\nabla:\gd G_M \to \gd^H(\pi_*G_P) \label{+230}
\eeq
of the exact sequence (\ref{+228}).
\end{definition}

In contrast with principal connections on smooth principal
bundles, principal superconnections on a $\wh H$-principal
superbundle need not exist.

A principal superconnection can be described  in terms of a ${\got
h}$-valued one-superform
\be
\om: \gd G_P\to G_P\ot {\got h}\cong \cV G_P,
\ee
on $\wh P$ called a superconnection form. Indeed, every splitting
$\nabla$ (\ref{+230}) defines the morphism of $G_P$-modules
\be
\wh \pi^*(\gd G_M)\to \wh\pi^*(\gd^H(\pi_*G_P))\cong \gd G_P
\ee
which splits the exact sequence
\be
0\to \cV G_P\to\gd G_P\to \wh\pi^*(\gd G_M)\to 0.
\ee
Therefore, there exists the exact sequence
\be
0\to \wh\pi^*(\gd G_M)\cV G_P\to\gd G_P\op\to^\om \cV G_P\to 0.
\ee

Let us note that, by analogy with associated smooth bundles, one
can introduce associated  superbundles and superconnections on
these superbundles. In particular, every supervector bundle of
fibre dimension $(r,s)$ is a superbundle associated with
$\wh{GL}(r|s;\La)$-principal superbundle \cite{bart}.

\section{Supermetric}

In gauge theory on a principal bundle $P\to X$ with a structure
Lie group $G$ reduced to its subgroup $H$, the corresponding
global section of the quotient bundle $P/H\to X$ is regarded as a
classical Higgs field \cite{book09,higgs}, e.g., a gravitational
field in gauge gravitation theory \cite{book09,iva,sard06}.

Let $\pi:P\to X$ be a principal smooth bundle with a structure Lie
group $G$. Let $H$ be a closed (consequently, Lie) subgroup of
$G$. Then $G\to G/H$ is an $H$-principal fiber bundle and, by the
well known theorem, $P$ is split into the composite fiber bundle
\mar{g1}\beq
P\ar^{\pi_H} P/H\ar X, \label{g1}
\eeq
where $P\to P/H$ is an $H$-principal bundle and $P/H\to X$ is a
$P$-associated bundle with the typical fiber $G/H$. One says that
the structure group $G$ of a principal bundle $P$ is reducible to
$H$ if there exists an $H$-principal subbundle of $P$. The
necessary and sufficient conditions of the reduction of a
structure group are stated by the well known theorem
\cite{book09,higgs}.

\begin{theorem} \label{g00} \mar{g00} There is one-to-one
correspondence $P^h=\pi_H^{-1}(h(X))$ between the reduced
$H$-principal subbundles $P^h$ of $P$ and the global sections $h$
of the quotient bundle $P/H\to X$.
\end{theorem}

As was mentioned above, sections of $P/H\to X$ are treated in
gauge theory as classical Higgs fields. For instance, let $P=LX$
be the $GL(n,\Bbb R)$-principal bundle of linear frames in the
tangent bundle $TX$ of $X$ (n=\di X). If $H=O(k,n-k)$, then a
global section of the quotient bundle $LX/O(k,n-k)$ is a
pseudo-Riemannian metric on $X$.

Our goal is the following extension of Theorem \ref{g00} to
principal superbundles \cite{sard08}.

\begin{theorem} \label{g20} \mar{g20}
Let $\wh P\to \wh M$ be a principal $G$-superbundle with a
structure $G$-Lie supergroup $\wh G$, and let $\wh H$ be a closed
$G$-Lie supersubgroup of $\wh G$ such that $\wh G\to \wh G/\wh H$
is a principal superbundle. There is one-to-one correspondence
between the principal $G$-supersubbundles of $\wh P$ with the
structure $G$-Lie supergroup $\wh H$ and the global sections of
the quotient superbundle $\wh P/\wh H\to \wh M$ with the typical
fiber $\wh G/\wh H$.
\end{theorem}

In order to proof Theorem \ref{g20}, it suffices to show that the
morphisms
\mar{g21}\beq
\wh P\ar \wh P/\wh H\ar \wh M \label{g21}
\eeq
form a composite $G$-superbundle. A key point is that underlying
spaces of $G$-supermanifolds are smooth real manifolds, but
possessing very particular transition functions and morphisms.
Therefore, the condition of local triviality of the quotient $\wh
G\to \wh G/\wh H$ is rather strong. However, it is satisfied in
the most interesting case for applications when $\wh G$ is a
supermatrix group and $\wh H$ is its Cartan supersubgroup. For
instance, let $\wh P=L\wh M$ be a principal superbundle of graded
frames in the tangent superspaces over a supermanifold $\wh M$ of
even-odd dimension $(n,2m)$. If its structure general linear
supergroup $\wh G=\wh{GL}(n|2m; \La)$ is reduced to the
orthgonal-symplectic supersubgroup $\wh H=\wh{OS}p(n|m;\La)$, one
can think of the corresponding global section of the quotient
bundle $L\wh M/\wh H\to \wh M$ as being a supermetric on $\wh M$.
Note that a Riemannian supermetric on graded manifolds has been
considered in a different way \cite{zirn}.

\begin{proof}
Let $\wh\pi:\wh P\to \wh P/\wh G$ be a principal superbundle with
a structure $G$-Lie group $\wh G$. Let $\wh i:\wh H\to \wh G$ be a
closed $G$-Lie supersubgroup of $\wh G$, i.e., $i: H\to G$ is a
closed Lie subgroup of the Lie group $G$. Since $H$ is a closed
subgroup of $G$, the latter is an $H$-principal fiber bundle $G\to
G/H$ \cite{ste}. However, $G/H$ need not possesses a
$G$-supermanifold structure. Let us assume that the action
\be
\wh\rho:\wh G\times \wh H\ar^{\id\times \wh i}\wh G\times \wh
G\ar^{\wh m}\wh G
\ee
of $\wh H$ on $\wh G$ by right multiplications defines the
quotient
\mar{g80}\beq
\wh\zeta:\wh G\to \wh G/\wh H \label{g80}
\eeq
which is a principal superbundle with the structure $G$-Lie
supergroup $\wh H$. In this case, the $G$-Lie supergroup $\wh G$
acts on the quotient supermanifold $\wh G/\wh H$ on the left by
the law
\be
\wh\varrho: \wh G\times \wh G/\wh H=\wh G\times \wh\zeta(\wh G)\to
(\wh\zeta\circ\wh m)(\wh G\times \wh G).
\ee
Given this action of $\wh G$ on $\wh G/\wh H$, we have a $\wh
P$-associated superbundle
\mar{g87}\beq
\wh\Si=(\wh P\times \wh G/\wh H)/\wh G \ar^{\wh\pi_\Si} \wh M
\label{g87}
\eeq
with the typical fiber $\wh G/\wh H$. Since
\be
\wh P/\wh H=((\wh P\times \wh G)/\wh G)/\wh H= (\wh P\times \wh
G/\wh H)/\wh G,
\ee
the superbundle $\wh\Si$ (\ref{g87}) is the quotient $(\wh P/\wh
H, \wh\pi_H)$ of $\wh P$ with respect to the right action
\be
\wh\rho\circ (\id\times \wh i):\wh P\times \wh H\ar \wh P\times
\wh G\ar \wh P
\ee
of the $G$-Lie supergroup $\wh H$. Let us show that this quotient
$\wh\pi_H:\wh P\to \wh P/\wh H$ is a principal superbundle with
the structure supergroup $\wh H$. Note that, by virtue of the
well-known theorem \cite{ste}, the underlying space $P$ of $\wh P$
is an $H$-principal bundle $\pi_H: P\to P/H$. Let
$\{V_\kappa,\wh\Psi_\kappa\}$ be an atlas of trivializations
\be
\wh\Psi_\kappa : (\zeta^{-1}(V_\kappa),
\ccG_G|_{\zeta^{-1}(V_\kappa)})\to
(V_\kappa,\ccG_{G/H}|_{V_\kappa})\times \wh H,
\ee
of the $\wh H$-principal bundle $\wh G\to \wh G/\wh H$, and let
$\{U_\al,\wh\psi_\al\}$ be an atlas of trivializations
\be
\wh\psi_\al: (\pi^{-1}(U_\al),\ccG_P|_{\pi^{-1}(U_\al)})\to
(U_\al,\ccG_M|_{U_\al})\times \wh G
\ee
of the $\wh G$-principal superbundle $\wh P\to \wh M$. Then we
have the $G$-isomorphisms
\mar{g85}\ben
&& \wh\psi_{\al\kappa}=(\id \times \wh\Psi_\kappa)\circ
\wh\psi_\al: (\psi^{-1}_\al(U_\al\times\zeta^{-1}(V_\kappa)),
\ccG_P|_{\psi^{-1}_\al(U_\al\times\zeta^{-1}(V_\kappa))})\to  \label{g85}\\
&& \qquad (U_\al,\ccG_M|_{U_\al})\times
(V_\kappa,\ccG_{G/H}|_{V_\kappa}) \times \wh H=(U_\al\times
V_\kappa, \ccG_M|_{U_\al}\wh\ot \ccG_{G/H}|_{V_\kappa})\times \wh
H. \nonumber
\een
For any $U_\al$, there exists a well-defined morphism
\be
&& \wh\Psi_\al:(\pi^{-1}(U_\al), \ccG_P|_{U_\al}) \to (U_\al\times
G/H, \ccG_M|_{U_\al}\wh\ot \ccG_{G/H})\times \wh H=\\
&& \qquad (U_\al,\ccG_M|_{U_\al})\times \wh G/\wh H\times\wh H
\ee
such that
\be
\wh \Psi_\al|_{\psi^{-1}_\al(U_\al\times\zeta^{-1}(V_\kappa))}
=\wh \psi_{\al\kappa}.
\ee
Let $\{U_\al,\wh\vf_\al\}$ be an atlas of trivializations
\be
\wh\vf_\al:
(\pi_\Si^{-1}(U_\al),\ccG_\Si|_{\pi^{-1}_\Si(U_\al)})\to
(U_\al,\ccG_M|_{U_\al})\times \wh G/\wh H
\ee
of the $\wh P$-associated superbundle $\wh P/\wh H\to \wh M$. Then
the morphisms
\be
(\wh\vf_\al^{-1}\times \id)\circ \wh\Psi_\al: (\pi^{-1}(U_\al),
\ccG_P|_{U_\al})\to (\pi_\Si^{-1}
(U_\al),\ccG_\Si|_{\pi^{-1}_\Si(U_\al)})\times\wh H
\ee
make up an atlas $\{\pi^{-1}_\Si(U_\al), (\wh\vf_\al^{-1}\times
\id)\circ \wh\Psi_\al\}$ of trivializations of the  $\wh
H$-principal superbundle $\wh P\to \wh P/\wh H$. As a consequence,
we obtain the composite superbundle (\ref{g21}). Now, let $\wh i_h
: \wh P_h\to \wh P$ be an $\wh H$-principal supersubbundle of the
principal superbundle $\wh P\to\wh M$. Then there exists a global
section $\wh h$ of the superbundle $\wh \Si\to \wh M$ such that
the image of $\wh P_h$ with respect to the morphism $\wh
\pi_H\circ\wh i_h$ coincides with the range of the section $\wh
h$. Conversely, given a global section $\wh h$ of the superbundle
$\wh \Si\to \wh M$, the inverse image $\wh\pi_H^{-1}(\wh h(\wh
M))$ is an $\wh H$-principal supersubbundle of $\wh P\to\wh M$.
\end{proof}

Let us show that, as was mentioned above, the condition of Theorem
\ref{g20} hold if $\wh H$ is the Cartan supersubgroup of a
supermatrix group $\wh G$, i.e., $\wh G$ is a $G$-Lie
supersubgroup of some general linear supergroup
$\wh{GL}(n|m;\La)$.

Recall that a Lie superalgebra $\wh{\got g}$ of an
$(n,m)$-dimensional $G$-Lie supergroup $\wh G$ is defined as a
$\La$-algebra of left-invariant supervector fields on $\wh G$,
i.e., derivations of its structure sheaf $\ccG_G$. A supervector
field $u$ is called left-invariant if
\be
(\id\ot u)\circ \wh m^* =\wh m^*\circ u.
\ee
Left-invariant supervector fields on $\wh G$ make up a Lie
$\La$-superalgebra. Being a superspace $B^{n|m}$, a Lie
superalgebra is provided with a structure of the standard
$G$-supermanifold $\wh B^{n+m,n+m}$. Its even part $\wh{\got
g}_0=\wh B^{n,m}$ is a Lie $\La_0$-algebra.

Let $\wh G$ be a matrix $G$-Lie supergroup. Then there is an
exponential map
\be
\xi(J)=\exp(J)=\op\sum_k \frac{1}{k'}J^k
\ee
of some open neighbourhood of the origin of the Lie algebra
$\wh{\got g}_0$ onto an open neighbourhood $U$ of the unit of $\wh
G$. This map is an $H^\infty$-morphism, which is trivially
extended to a $G$-morphism.

Let $\wh H$ be a Cartan supersubgroup of $\wh G$, i.e., the even
part $\wh{\got h}_0$ of the Lie superalgebra $\wh{\got h}$ of $wh
H$ is a Cartan subalgebra of the Lie algebra $\wh\ccG_0$, i.e.,
\be
\wh\ccG_0=\wh{\got f}_0 +\wh{\got h}_0, \qquad [\wh{\got
f}_0,\wh{\got f}_0]\subset \wh{\got h}_0, \qquad [\wh{\got
f}_0,\wh{\got h}_0]\subset \wh{\got f}_0.
\ee
Then there exists an open neighbourhood, say again $\wh U$, of the
unit of $\wh G$ such that any element $g$ of $\wh U$ is uniquely
brought into the form
\be
g=\exp(F)\exp(I), \qquad F\in \wh{\got f}_0, \qquad I\in \wh{\got
h}_0.
\ee
Then the open set $\wh U_H=\wh m(\wh U\times \wh H)$ is
$G$-isomorphic to the direct product $\xi(\xi^{-1}(U)\cap\wh{\got
f}_0)\times\wh H$. This product provides a trivialization of an
open neighbourhood of the unit of $\wh G$. Acting on this
trivialization by left translations $\wh L_g$, $g\in\wh G$, one
obtains an atlas of a principal superbundle $\wh G\to\wh H$.

For instance, let us consider a superspace $B^{n|2m}$, coordinated
by $(x^a,y^i,\ol y^i)$, and the general linear supergroup
$\wh{GL}(n|2m;\La)$ of its automorphisms. Let $B^{n|2m}$ be
provided with the $\La$-valued bilinear form
\mar{g90}\beq
\om=\op\sum_{i=1}^n (x^i x'^i) + \op\sum_{j=1}^m (y^j\ol y'^j- \ol
y^j y'^j). \label{g90}
\eeq
The supermatrices (\ref{+200}) preserving this bilinear form make
up the orthogonal-symplectic supergroup $\wh{OS}p(n|m;\La)$
\cite{fuks}. It is a Cartan subgroup of $\wh{GL}(n|2m;\La)$. Then
one can think of the quotient
$\wh{GL}(n|2m;\La)/\wh{OS}p(n|m;\La)$ as being a supermanifold of
$\La$-valued bilinear forms on $B^{n|2m}$ which are brought into
the form (\ref{g90}) by general linear supertransformations.

Let $\wh M$ be $G$-supermanifold of dimension $(n,2m)$ and $T\wh
M$ its tangent superbundle. Let $L\wh M$ be an associated
principal superbundle. Let us assume that its structure supergroup
$\wh{GL}(n|2m;\La)$ is reduced to the supersubgroup
$\wh{OS}p(n|m;\La)$. Then by virtue of Theorem \ref{g20}, there
exists a global section $h$ of the quotient $L\wh
M/\wh{OS}p(n|m;\La)\to \wh M$ which can be regarded as a
supermetric on a supermanifold $\wh M$.

Note that, bearing in mind physical applications, one can treat
the bilinear form (\ref{g90}) as {\it sui generis} superextension
of the Euclidean metric on the body $\Bbb R^n=\si(B^{n|m})$ of the
superspace $B^{n|m}$. However, the body of a supermanifold is
ill-defined in general \cite{caten}.

\section{Graded principal bundles}

Graded principal bundles and connections on these bundles  can be
studied similarly to principal superbundles and principal
superconnections, though the theory of graded principal bundles
preceded that of principal superbundles \cite{alm,kost77}.
Therefore, we will touch on only a few elements of the graded
bundle technique (see, e.g. \cite{stavr} for a detailed
exposition).

Let $(Z,\gA)$ be a graded manifold of dimension $(n,m)$. A useful
object in the graded manifold theory, not mentioned above, is the
 finite dual $\gA(Z)^\circ$ of the algebra $\gA(Z)$ which
consists of elements $a$ of the dual $\gA(Z)^*$ vanishing on an
ideal of $\gA(Z)$ of finite codimension. This is a graded
commutative coalgebra with the comultiplication
\be
(\Delta^\circ(a))(f\ot f')=a(ff'), \qquad  f,f'\in \gA(Z),
\ee
and the counit
\be
\e^\circ(a)= a(1_\gA).
\ee
In particular, $\gA(Z)^\circ$ includes the evaluation elements
$\dl_z$ such that
\be
\dl_z(f)=(\si(f))(z).
\ee
Given an evaluation element $\dl_z$, elements $u\in\gA(Z)^\circ$
are called primitive elements with respect to $\dl_z$ if they obey
the relation
\mar{+234}\beq
\Delta^\circ(v)= u\ot\dl_z +\dl_z\ot u. \label{+234}
\eeq
These elements are derivations of $\gA(Z)$ at $z$, i.e.,
\be
u(ff')=(uf)(\dl_zf') +(-1)^{[u][f]}(\dl_zf)(uf').
\ee

\begin{definition}
A graded Lie group $(G,\ccG)$ is defined as a graded manifold such
that $G$ is an ordinary Lie group, the algebra $\ccG(G)$ is a
graded Hopf algebra $(\Delta, \e, S)$, and the algebra epimorphism
$\si:\ccG(G)\to C^\infty(G)$ is a morphism of graded Hopf
algebras.
\end{definition}

  One can show that $\ccG(G)^\circ$ is also equipped
with the structure of a Hopf algebra with the multiplication law
\mar{+236}\beq
a\star b=(a\ot b)\circ\Delta, \qquad  a,b\in\ccG(G)^\circ.
\label{+236}
\eeq
With respect to this multiplication, the evaluation elements
$\dl_g$, $\ccG\in G$, constitute a group
$\dl_g\star\dl_{g'}=\dl_{gg'}$ isomorphic to $G$. Therefore, they
are also called group-like elements. It is readily observed that
the set of primitive elements of $\ccG(G)^\circ$ with respect to
$\dl_e$, i.e., the tangent space $T_e(G,\ccG)$ is a  Lie
superalgebra with respect to the multiplication (\ref{+236}). It
is called the Lie superalgebra ${\got g}$ of the graded Lie group
$(G,\ccG)$.

One says that a graded Lie group $(G,\ccG)$ acts on a graded
manifold $(Z,\gA)$ on the right if there exists a morphism
\be
(\vf,\Phi):(Z,\gA)\times(G,\ccG)\to (Z,\gA)
\ee
such that the corresponding algebra morphism
\be
\Phi:\gA(Z)\to \gA(Z)\ot \ccG(G)
\ee
defines a structure of a right $\ccG(G)$-comodule on $\gA(Z)$,
i.e.,
\be
(\Id\ot \Delta)\circ\Phi=(\Phi\ot\Id)\circ\Phi, \qquad
(\Id\ot\e)\circ\Phi=\Id.
\ee
For a right action $(\vf,\Phi)$ and for each element
$a\in\ccG(G)^\circ$, one can introduce the linear map
\mar{+235}\beq
  \Phi_a=(\Id\ot a)\circ\Phi: \gA(Z)\to\gA(Z). \label{+235}
\eeq
In particular, if $a$ is a primitive element with respect to
$\dl_e$, then $\Phi_a\in\gd \gA(Z)$.

Let us consider a right action of $(G,\ccG)$ on itself.  If
$\Phi=\Delta$ and $a=\dl_g$ is a group-like element, then $\Phi_a$
(\ref{+235}) is a homogeneous graded algebra isomorphism of degree
zero which  corresponds to the right translation $G\to Gg$. If
$a\in{\got g}$, then $\Phi_a$ is a derivation of $\ccG(G)$. Given
a basis $\{u_i\}$ for ${\got g}$, the derivations $\Phi_{u_i}$
constitute the global basis for $\gd\ccG(G)$, i.e., $\gd\ccG(G)$
is a free left $\ccG(G)$-module. In particular, there is the
decomposition
\be
&&\ccG(G)=\ccG'(G)\oplus_R\ccG''(G),\\
&& \ccG'(G)=\{ f\in\ccG(G)\, :\, \Phi_u(f)=0, \,\,  u\in {\got
g}_0\},\\
&& \ccG''(G)=\{ f\in\ccG(G)\, :\, \Phi_u(f)=0, \,\,  u\in {\got
g}_1\}.
\ee
Since $\ccG'(G)\cong C^\infty(G)$, one finds that every graded Lie
group $(G,\ccG)$ is the sheaf of sections of some trivial exterior
bundle $G\times {\got g}_1^*\to G$ \cite{alm,boy,kost77}.

Let us turn now to the notion of a graded principal bundle. A
right action $(\vf,\Phi)$ of $(G,\ccG)$ on $(Z,\gA)$ is called
free if, for each $z\in Z$, the morphism
\be
\Phi_z:\gA(Z)\to \ccG(G)
\ee
is such that the dual morphism
\be
\Phi_{z*}:\ccG(G)^\circ \to \gA(Z)^\circ
\ee
is injective.

A right action $(\vf,\Phi)$ of $(G,\ccG)$ on $(Z,\gA)$ is called
regular if the morphism
\be
(\vf\times\pr_1)\circ\Delta: (Z,\gA)\times (G,\ccG)\to
(Z,\gA)\times (Z,\gA)
\ee
defines a closed graded submanifold of $(Z,\gA)\times (Z,\gA)$.

\begin{remark}
Let us note that $(Z',\gA')$ is said to be a graded submanifold of
$(Z,\gA)$ if there exists a morphism $(Z',\gA')\to (Z,\gA)$ such
that the corresponding morphism $\gA'(Z')^\circ\to \gA(Z)^\circ$
is an inclusion. A graded submanifold is called closed if ${\rm
dim}\,(Z',\gA')<{\rm dim}\,(Z,\gA)$.
\end{remark}

Then we come to the following variant of the well-known theorem on
the quotient of a graded manifold \cite{alm,stavr}.

\begin{theorem}
A right action $(\vf,\Phi)$ of $(G,\ccG)$ on $(Z,\gA)$ is regular
if and only if the quotient $(Z/G,\gA/\ccG)$ is a graded manifold,
i.e., there exists an epimorphism of graded manifolds $(Z,\gA)\to
(Z/G,\gA/\ccG)$ compatible with the surjection $Z\to Z/G$.
\end{theorem}

In view of this Theorem,  a graded principal bundle $(P,\gA)$ can
be defined as a locally trivial submersion
\be
(P,\gA)\to (P/G,\gA/\ccG)
\ee
with respect to the right regular free action of $(G,\ccG)$ on
$(P,\gA)$.  In an equivalent way, one can say that a graded
principal bundle is a graded manifold  $(P,\gA)$ together with a
free right action of a graded Lie group $(G,\ccG)$ on  $(P,\gA)$
such that the quotient $(P/G,\gA/\ccG)$ is a graded manifold and
the natural surjection
\be
(P,\gA)\to (P/G,\gA/\ccG)
\ee
is a submersion. Obviously, $P\to P/G$ is an ordinary
$G$-principal bundle.

A graded principal connection on a graded $(G,\ccG)$-principal
bundle $(P,\gA)\to (X,{\got B})$ can be introduced similarly to
  a  superconnection on a principal superbundle. This is defined as a
$(G,\ccG)$-invariant splitting of the sheaf $\gd\gA$, and is
represented by a ${\got g}$-valued graded connection form on
$(P,\gA)$ \cite{stavr}.

\begin{remark}
In an alternative way, one can define graded connections on a
graded bundle $(Z,\gA)\to (X,{\got B})$ as sections $\G$ of the
jet graded bundle
\be
J^1(Z/X)\to (Z,\gA)
\ee
of sections of  $(Z,\gA)\to (X,{\got B})$ \cite{alm}, which is
also a graded manifold \cite{rup}. In the case of a
$(G,\ccG)$-principal graded bundle, these sections $\G$ are
required to be $(G,\ccG)$-equivariant.
\end{remark}

\end{document}